\documentclass[11pt]{article}

\usepackage{amssymb}
\usepackage{subfigure}
\usepackage{graphicx}
\usepackage{epsfig}
\usepackage{latexsym}
\usepackage{amsmath}

\newcommand{\ignore}[1]{}
\newcommand{\remove}[1]{}
\newcommand{\comment}[1]{}

 \newtheorem{theorem}{Theorem}[section]

 \newtheorem{corollary}[theorem]{Corollary}
 \newtheorem{lemma}[theorem]{Lemma}

 \newtheorem{definition}{Definition}[section]

 \newcommand{\qed}{\nopagebreak \hfill $\Box$}
 \newenvironment{proof}{\par \noindent {\bf Proof}:}{\qed \par}

\begin{document}

%\frontmatter          % for the preliminaries
%
%\pagestyle{headings}

\newcommand{\fig}[3] %usage:\fig{file}{label}{caption}

\newcommand{\aline}{\centerline{\rule{8.5in}{0.1mm}}}
\newcommand{\hl}{\hat{\lambda}}
\newcommand{\eps}{\varepsilon}
\newcommand{\vp}{\varphi}
\newcommand{\cm}[1]{}
\newcommand{\lset}{\left \{}
\newcommand{\rset}{\right \}}
\newcommand{\ceil}[1]{\lceil #1 \rceil }

\def \From{From }
\newcommand{\ca}{\mbox{$\cal A$}}
\newcommand{\Sa}{\mbox{$S_{\cal A}$}}

\newcommand{\ci}{\mbox{${\cal I}$}}
\newcommand{\co}{\mbox{$\cal O$}}
\newcommand{\lao}{\lambda_1}
\newcommand{\boa}{\bar{O}_A}
\newcommand{\oo}{{\cal O}_1}
\newcommand{\hoo}{\hat{{\cal O}_1}}
\newcommand{\lat}{\lambda_2}
\newcommand{\tlat}{\tilde{\lambda}_2}
\newcommand{\tw}{\tilde{W}}
\newcommand{\te}{\tilde{\eps}}
\newcommand{\ot}{{\cal O}_2}
\newcommand{\las}{\lambda^*}
\newcommand{\wos}{W^*_1}
\newcommand{\wo}{W_1}
\newcommand{\wts}{W^*_2}
\newcommand{\wt}{W_2}
\newcommand{\ha}{\hat{A}}
\newcommand{\hao}{\hat{A_1}}
\newcommand{\haso}{\hat{A^*_1}}
\newcommand{\ao}{A_1}
\newcommand{\at}{A_2}
\newcommand{\aos}{A^*_1}
\newcommand{\ats}{A^*_2}
\newcommand{\amha}{A_1 \setminus \hat{A_1}}
\newcommand{\asmha}{A^*_1 \setminus \hat{A^*_1}}
\newcommand{\azj}{A_{0,j}}
\newcommand{\aj}{A_j}
\newcommand{\oomh}{{\oo} \setminus {\hoo}}
\newcommand{\ko}{k_1}
\newcommand{\kos}{k^*_1}
\newcommand{\xos}{x^*_1}
\newcommand{\orpim}{{\cal O}_b}
\newcommand{\kt}{k_1}
\newcommand{\kts}{k^*_2}
\newcommand{\xts}{x^*_1}
\newcommand{\rpim}{R(\Pi,b)}
\newcommand{\prob}[1]{\Pr\left[\mbox{#1}\right]}
\newcommand{\argmax}{\mbox{argmax}}
\newcommand{\argmin}{\mbox{argmin}}

\newcommand{\mO} {\mathcal{O}}
\newcommand{\cA} {\mathcal{A}}
\newcommand{\cI} {\mathcal{I}}

\newcommand{\myparagraph}[1]{\par\smallskip\par\noindent{\bf{}#1:~}}
\newcommand{\negA}{\vspace{-0.07in}}  %was 0.1
\newcommand{\negB}{\vspace{-0.12in}}  %was 0.14
\newcommand{\negC}{\vspace{-0.18in}}  %was 0.18
\newcommand{\posA}{\vspace{0.08in}}
\newcommand{\mysection}[1]{\section{#1}}
\newcommand{\mysubsection}[1]{\subsection{#1}}
\newcommand{\mysubsubsection}[1]{\subsubsection{#1}}

\newenvironment{sketch}{\noindent {\bf Proof sketch: } }%

\newtheorem{thm}{Theorem}[section]
\newtheorem{cor}{Corollary}
\newtheorem{alg}{Algorithm}

\title{On Lagrangian Relaxation and Reoptimization Problems
\footnote{A preliminary version of this paper appeared in the Proceedings of the
6th Workshop on Approximation and Online Algorithms, Germany, September 2008.
Research supported
 by the Israel Science Foundation
(grant number 1574/10), and
by the Ministry of Trade and Industry MAGNET program through the NEGEV
 Consortium.}}

\author{
Ariel Kulik \thanks{Computer Science Department, Technion,
Haifa 32000, Israel. \mbox{E-mail: {\tt \{kulik,hadas,galtamir\}@cs.technion.ac.il}}}
\and Hadas Shachnai\thanks{Corresponding author.}
\and Gal Tamir}

\maketitle

\begin{abstract}

We prove a general result demonstrating the power of Lagrangian relaxation in solving constrained maximization problems with arbitrary objective functions.
This yields a unified approach for solving a wide class of {\em subset selection} problems with linear constraints.
Given a problem in this class and some small $\eps \in (0,1)$,
we show that if there exists an $r$-approximation algorithm for the Lagrangian relaxation of the problem, for some $r \in (0,1)$, then our technique achieves a ratio of $\frac{r}{r+1} -\! \eps$ to the optimal, and this ratio is tight.

The number of calls to the $r$-approximation algorithm, used by our algorithms, is {\em linear} in the input size and in $\log (1 / \eps)$ for inputs with cardinality constraint, and polynomial in the input size and in $\log (1 / \eps)$ for inputs with arbitrary linear constraint. Using the technique we obtain (re)approximation algorithms for natural (reoptimization) variants of classic subset selection problems, including real-time scheduling,
the {\em maximum generalized assignment problem (GAP)} and maximum weight
independent set.
\comment{
In addition using that technique we obtain a general $(1, \alpha)$-reoptimization scheme for subset selection problems.
In the Reoptimization version of a problem we are require to resolve the problem when a previous solution is presented with minimal changes (transition cost).
Each new solution is measured by tow objective functions, the original problem's objective profit function $p(x)$ and the transition cost function $\delta(x)$.
When approximating Reoptimization problem we compare our solution to the optimal solution $OPT$
(the optimal solution has the minimal transition cost among all the optimal solution of $\Pi$), $OPT=\{S|p(S)=p(\mO), \delta(S)=\argmin_{p(S)=p(\mO)}{\delta(S)}\}$
thus a solution $S$ is a $(r_1, r_2 )$-reapproximation if its objective functions approximated compared to the optimal solution $OPT$,
that is $\delta(S) \leq r_1\cdot \delta (OPT)$
and $p(S) \geq r_2 \cdot p(OPT)$. We use that scheme to obtain a reapproximation algorithm to the "Cloud Provider problem" and the "Video on Demand problem".
}
\end{abstract}

%\smallskip
%\noindent {\bf Key Words and Phrases:}
%Lagrangian relaxation,
%constrained maximization, subset selection, approximation algorithms

%\end{titlepage}
%\newpage
%\setcounter{page}{1}

%\pagenumbering{arabic}
%\pagestyle{plain}
%\pagestyle{headings}

\section{Introduction}

%\subsection{Subset Selection Problems}
Lagrangian relaxation is a fundamental technique in combinatorial optimization. It has been used extensively in the design of approximation algorithms for a variety of problems (see e.g.,\cite{GR96,G96,KR00,JV01,KPS06,BBGS08} and a comprehensive survey in \cite{M07}).
%for solving constrained optimization problems.
%Jain and Vazirani developed in \cite{JV01} (see also \cite{Ga-96})
%a general framework
%of Jain and Vazirani \cite{Va-01}[pp.  250-251] (see also \cite{Ga-96}),
%for using Lagrangian relaxation to derive
%approximation algorithms (see also \cite{Ga-96}).
%This led to efficient approximations for a variety of
%constrained {\em minimization??} problems (see, e.g., in \cite{??} and
%comprehensive surveys in \cite{Mestre-thesis+survey}).
%Building on the framework of \cite{JV01},
In this paper we prove a general result demonstrating the power of
Lagrangian relaxation
in solving constrained maximization problems of the following form.
Given a universe $U$, a weight function $w:U \rightarrow \mathbb{R^+}$,
a function $f:U \rightarrow \mathbb{N}$ and an integer $L \geq 1$, we
want to solve
%the following problem.
%$\Pi$ given by
\begin{eqnarray}
\label{prb:general_constraint}
\Pi:& {\max}_{s \in U} & f(s) \\
%\Pi:& {maximize} & f(s) \\
\nonumber
&\mbox{subject to:} & w(s) \leq L.
%\\ && s \in U
\end{eqnarray}
%(We note that the size of $U$ may be exponentially large, thus we cannot omit
%initially all elements $s \in U$ for which $w_s > L$.)
We solve $\Pi$ by finding an
efficient solution for the Lagrangian relaxation of $\Pi$, given by
\begin{equation}
%\[
\label{prb:general_lagrange}
\Pi (\lambda): ~\max_{s\in U}  {f(s) - \lambda\cdot w(s),}
%\]
\end{equation}
for some $\lambda \geq 0$.

A traditional approach for using Lagrangian relaxation in
approximation algorithms (see, e.g., \cite{G96,JV01,BBGS08})
%The approach
is based on
%Initially, we find
%consists of
initially finding two solutions, ${SOL}_1$, ${SOL}_2$, for
$\Pi (\lambda_1), \Pi (\lambda_2)$, respectively,
for some $\lambda_1, \lambda_2$, such that each of the solutions is an
approximation for the corresponding Lagrangian relaxation; while one of
these solutions is feasible for $\Pi$ (i.e., satisfies the weight constraint),
the other is not.
A main challenge is then to find a way to {\em combine}
${SOL}_1$ and ${SOL}_2$
%the two solutions
to a feasible solution that yields a good approximation for $\Pi$.
%, by {\em combining}
We prove (in Theorem \ref{lemma:lagrangian_relaxation}) a general result,
which allows to obtain a solution for $\Pi$ based on {\em one} of the
solutions only. In particular, we show that, with appropriate selection of the
parameters $\lambda_1, \lambda_2$
in the Lagrangian relaxation, we can obtain solutions
${SOL}_1$, ${SOL}_2$
%$s_1, s_2$,
such that one of them yields an
efficient approximation for our original problem $\Pi$.
%$s_1$, or $s_2$.
The resulting
%approximation
technique leads to fast and simple approximation algorithms for
a wide class of {\em subset selection} problems with linear constraints.

\subsection{Subset Selection Problems}
Subset selection problems form a large class
%of polynomially solvable as well as
%problems as maximum matching and
encompassing such NP-hard problems as %bandwidth allocation,
real-time scheduling,
the generalized assignment problem (GAP) and maximum weight independent set,
among others. In these problems,
a subset of elements satisfying certain
properties
%constraints
needs to be selected out of a universe, so as to maximize some
objective function.\footnote{We give a formal definition in Section
\ref{sec:rounding}.}
%In this paper we consider the following
%constarined subset selection problems.
We apply our general technique to obtain efficient approximate
solutions for the following natural variants of some classic subset selection
problems.
%\begin{enumerate}
%\item[(i)]

\myparagraph{Budgeted Real Time Scheduling (BRS)}
The input is a set
%HADAS- text has been reordered due to display problem.
${\cal A} = \{ A_1, \ldots , A_m \}$
of \emph{activities},
where each activity consists of a set
of \emph{instances};
an instance $\cI \in A_i$ is defined
by a half open time interval $[s(\cI),e(\cI))$ in which the instance
can be scheduled ($s(\cI)$ is the start time, and $e(\cI)$ is
the end time), a cost
$c(\cI)\in \mathbb{N}$, and a profit $p(\cI)\in \mathbb{N}$.
A schedule is \emph{feasible} if it contains
at most one instance of each activity, and for any $t \geq 0$, at
%at any point of
%time $t$, there is at
most one instance is scheduled at time $t$.
%the total resource requirement of instances scheduled at time $t$
%is bounded by $1$.
The goal is to find a feasible schedule, in which
the total cost of all the scheduled instances
is bounded by a given budget $L \in \mathbb{N}$, and
the total profit of the scheduled instances is maximized.
%(subject to this constraints).
{\em Budgeted continuous real-time scheduling (BCRS)}
is a variant of this problem where each instance is associated
with a {\em time window}
$\cI = [s(\cI),e(\cI))$
%$\mathcal{T}= [s(\mathcal{T}),e(\mathcal{T}))$
and length $\ell(\cI)$.
%, instead of time interval.
An instance $\cI$ can be scheduled at any time interval
$[\tau,\tau+\ell(\cI))$, such that
$s(\cI) \leq \tau \leq e(\cI) - \ell(\cI))$.
%$\tau$ is inside the time window associated with $\cI$.
%The non-budgeted version of this problem is
%discussed in \cite{BB00}, among with other variants and
%specific cases of the problem. An approximation algorithm
%was already given in \cite{NST07}, however, the results
%at that paper are incorrect.
BRS and BCRS arise in many scenarios in which we need to schedule activities
subject to resource constraints, e.g.,
storage requirements for the outputs of the activities.

\myparagraph{Budgeted Generalized Assignment Problem (BGAP)}
The input is a set of bins
(of arbitrary capacities)
%(possibly different) capacity constraints,
and a set of items, where
each item has a size, a value and a packing cost
for each bin. Also, we are given a budget $L \geq 0$.
The goal is to pack in the bins a feasible subset of items of maximum value,
such that the total packing cost is at most $L$.
%subject to the capacity constraints,
%such that the total cost for packing the selected items is at most $L$.
%\footnote{We assume
%all parameters in the input of BGAP to be integral.}
%The non-budgeted variant of
%this problem (GAP) is discussed in \cite{FG06}, for which
%a $(1-e^{-1}-\eps)$ approximation is obtained.
BGAP arises in many real-life scenarios (e.g., inventory
planning with delivery costs).

\myparagraph{Budgeted Maximum Weight Independent Set (BWIS)}
Given a budget $L$ and a graph $G=(V,E)$, where each
vertex $v\in V$ has an associated  profit $p_v$ (or, \emph
{weight}) and associated  cost $c_v$, choose a subset $V' \subseteq V$
such that $V'$  is an \emph{independent set} (i.e., for any $e=(v,u)\in E$, $v \notin V'$ or $u \notin V'$),
the total cost of vertices in $V'$, given by
$\sum_{v\in V'} c_v$, is bounded by $L$,
and the total profit of $V'$, $\sum_{v\in V'} p_v$, is maximized. BWIS is a generalization of the classical {\em maximum independent
set (IS)} and {\em maximum weight independent set (WIS)} problems.

\subsection{Combinatorial Reoptimization}
Traditional combinatorial optimization problems require finding
solutions for a single instance. However, many of the real-life
scenarios motivating these problems involve systems that change
over time. Thus, throughout the continuous operation of
such a system, it is required to compute solutions for new problem
instances, derived from previous instances. Moreover, since there is
some cost associated with the transition from one solution to
another, the solution for the new instance must be {\em close} to the former solution
(under certain distance measure).

Solving the resulting {\em reoptimization} problem involves two
challenges:
$(i)$ {\em computing} an optimal (or close to the optimal) solution for
a new instance, and
$(ii)$ efficiently {\em converting} a given solution to a new one (we give below the formal definitions).
Indeed, due to the transition costs, we seek for the modified instance an efficient solution
which can be reached at low cost.
In that sense, the given initial solution plays a restrictive role, rather than serve as guidance to the algorithm.

For example, consider a cloud provider that runs virtual machines for its customers.
The provider has a set of servers (hypervisors); each server has a limited resource
available for use by the virtual machines that it hosts.
%shared by
%the virtual machines hosted by the server.
%allowing to host some virtual machines.
The demand for a particular virtual machine changes over time, with a corresponding change in its resource consumption.
This may require to migrate some virtual machines among the servers,
% which incurs a predetermined cost,
so as to keep the total demand
on each server bounded by its resource availability.
Thus, given an initial assignment of virtual machines to the servers,
%following certain changes in resource consumption,
the provider has to find a new assignment which maximizes the profit from serving customers,
% of virtual machines to the servers,
%with the goal of maximizing the profit from customer service for this assignment,
while minimizing the total migration cost.

\comment{
	\myparagraph{\bf Applications} Reoptimization problems
	naturally 	arise
	in many real-life scenarios. Indeed, planned or unanticipated
	changes occur over time in almost any system. It is then required
	to respond to these changes quickly and efficiently. Ideally, the
	response should maintain high performance while affecting only a
	small portion of the system. In \cite{full} we give a
	detailed description of some of the applications for which our
	reoptimization model fits well. This includes storage systems for
	VoD services, communication services and other network problems,
	stock trading, production planning and vehicle routing.
}
%%%%%%%%%%%%%%%%%%%%%%%%%%%%%%%%%%%%%%%%%%%%%%%%%%%%%%%%%%%%%%%%%%%%%
%%%%%%%%%%%%%%%%%%%%%%%%%%%%%%%%%%%%%%%%%%%%%%%%%%%%%%%%%%%%%%%%%%%%%
\subsubsection{Definitions and Notation}

A reoptimization problem $R(\Pi)$ is comprised of three elements: a universe $U$ of
all feasible solutions, a profit $p: U \rightarrow \mathbb{R^+}$
and a transition cost $\delta : U \rightarrow \mathbb{R^+}$.

In solving the reoptimization version of a
maximization problem, the objective is to find a solution
that maximizes the profit, while minimizing the transition cost.
In the original optimization problem, we only need to maximize the profit (e.g. $\max\{p(s)|s\in U\}$). We denote this problem by $\Pi$ and refer to it as the {\em base problem}.
The transition cost for a particular instance of the
reoptimization problem represents the
cost of moving from an existing solution to a given new solution.
(for example, in the cloud provider scenario, this can be the number of virtual machines that
need to migrate to reach a new assignment).

\comment
{
In the following we formally define the model
for combinatorial
reoptimization. Given an optimization problem $\Pi$ , let $I_0$ be
an input for $\Pi$, and let ${\cal C_{I}}_0 = \{ C_{I_0}^1,
C_{I_0}^2, \ldots \}$ be the set of configurations corresponding to
the solution space of $\Pi$ for ${I_0}$.\footnote{A configuration
can be any representation of a (partial) solution for $\Pi$.} Each
configuration $C_{I_0}^j \in {\cal C_{I}}_0$ has some value
$val(C_{I_0}^j)$.

In the reoptimization problem, $R(\Pi)$,
an instance $I_R$ contains an initial configuration for the original problem $C_{I_0}^j \in {\cal C_{I}}_0$ of an initial instance ${I_0}$, and a
new instance $I$ derived from $I_0$ by
%various
admissible operation, e.g, addition or removal of elements,
changes in element parameters etc.

For any element $i \in I$ the transition cost of this element
$\delta(i)$ is derived directly from the initial configuration $C_{I_0}^j$.
Some time we even omit the initial configuration $C_{I_0}^j$ from $I_R$ and leave it with the transition cost function $\delta(i)$.

This representation keeps the input description more compact.
The primary goal is to find an optimal solution
for the problem $\Pi$ on the instance $I$.
Among all configurations with an optimal $val(C_{I}^k)$
value, we seek a configuration $C_{I}^{*}$ for which the
total transition cost, given by $\sum_{i \in
I \cap C_{I}^{*} }\delta(i)$ is minimized.
}

In the following, we show how the above notation can be used to describe
a reoptimization version of the {\em maximum spanning tree (MAX-ST)} problem. Denote this problem by $\Pi$.
Let $G_0=(V_0,E_0)$ be a weighted graph, and let $T_0=(V_0,E_{T_0})$
be a MAX-ST for $G_0$. Let $G=(V,E)$ be a graph derived from $G_0$ by
adding or removing vertices and/or edges, and by (possibly) changing the weights of edges.
Let $T=(V,E_{T})$ be a MAX-ST for $G$.
For every edge $e \in E_{T} $, we are given the cost $\delta(e)$ of
adding $e$ to the new solution (an example of transition cost
$\delta(e)$ can be: $\delta(e) = 0$ if $e \in E_0$; otherwise,
$\delta(e) = 1$).
The goal in the reoptimization problem $R($MAX-ST$)$ is to find a MAX-ST of $G$ with minimal total transition cost.
The formal representation of $R(\Pi)$ is $U=\{\mbox{all spanning trees of $G$}\}$,
$p(T)= \sum_{e\in T} w(e)$, where $w$ is the weight on the graph edges, and $\delta(T)= \sum_{e\in T} \delta(e)$.
A polynomial time algorithm for $R($MAX-ST$)$ is given in \cite{STT12}.
%\footnote{Clearly, if there is a unique
% MST for $G$, it suffices to find
%this MST.}
\comment{
The input for the reoptimization problem, $I_R$, contains both the
new instance, $I$, and the transition costs $\delta$ (that may be
encoded in  different ways). Note that $I_R$ does not include the
initial configuration $I_0$ since, apart from determining the
transition costs, it has no effect on the reoptimization problem.
}

It is worth noting that the input for the reoptimization problem contains only the new instance of the problem, but not the existing state
($T_0$ in the MAX-ST example), as the current state is reflected by the
transition cost $\delta$.

%%%%%%%%%%%%%%%%%%%%%%%%%%%%%%%%%%%%%%%%%%%%%%%%%%
\myparagraph{Approximate Reoptimization}
When the problem $\Pi$ is NP-hard, or when the reoptimization problem
$R(\Pi)$ is NP-hard,\footnote{As shown in \cite{STT12}, it may be that none, both, or only $R(\Pi)$ is NP-hard.}
we seek approximate solutions.
The goal is to find a good solution for the new instance, while keeping a low transition cost from the initial solution (or, {\em configuration}) to the new one.
Formally, denote by ${\cal O}$ an optimal solution for $\Pi$, that is, $p({\cal O})=\max_{s \in U} (p(s))$.
Denote by $OPT$ a solution for $R(\Pi)$ having the minimal transition cost among all the solutions that have a
 profit $p({\cal O})$. Formally, $\delta(OPT)=\argmin_{s\in U | p(s)= p({\cal O })} (\delta(s))$.
We now define the notion of {\em reapproximation} algorithm.\footnote{We refer the reader to \cite{STT12} for further details.}

\begin{definition}
 \label{def:weak_reapprox}
For $r_1\geq 1$ and $r_2 \in (0,1]$
a solution $s$ is an $(r_1, r_2 )$-reapproximation  for $R(\Pi)$ if it satisfies $(i)$ $\delta(s) \leq r_1\cdot \delta (OPT)$,
and $(ii)$ $p(s) \geq r_2 \cdot p(OPT)$.
\end{definition}

\comment{
Formally,
\begin{definition}
 \label{def:weak_reapprox}

An algorithm ${\cal A}$ yields an $(r_1, r_2)$-reapproximation for $R(\Pi)$, for $r_2,r_1 \geq 1$,
For any instance $I$ of a reoptimization problem $R(\Pi)$,
if its transition costs satisfies $delta (s) \leq r_1 \cdot \delta (OPT)$, ans its total profit is $p(s) \geq r_2 \cdot p(OPT)$.
 \end{definition}
}

In this paper we develop a framework that enables to obtain $(1,\alpha)$-reapproximation algorithms for a wide
class of subset selection problems, where $\alpha \in (0,1)$.
While some of the resulting approximation ratios may be improved, by applying problem-specific approximation
techniques (see Section \ref{sec:related_work}), our
framework is of interest due to its simplicity and generality.
We demonstrate the usefulness of our framework in solving the
following reoptimization problem.

\myparagraph{The Surgery Room Allocation Problem (SRAP)}
In a hospital, a surgery room is a vital resource.
Operations are scheduled by the severity of patient illness;
however, operation schedules tend to change due to sudden changes in patients' condition, the arrival
of new patients requiring urgent treatment, or the unexpected absence of senior staff members.
Schedule changes
involve some costs, e.g., due to the need to rearrange the equipment,
or to change the staff members taking care of the patients, as well as their individual schedules.

There is also a profit accrued from each operation. Indeed, some
operations are more profitable than others, e.g., due to the coverage received from insurance companies,
or due to higher charges in case the operation is scheduled after work hours.

Formally, suppose that the initial input, $I_0$, consists of $n_0$ patients. Each patient $j$ is associated with a set $A_{0,j}$  of possible time intervals in which
$j$ can be scheduled for operation. An interval ${\cal I} \in A_{0,j}$ is a half open time interval $[s_0({\cal I}), e_0({\cal I})$, where
$s_0({\cal I}) \leq e_0({\cal I})$. Each interval ${\ci} \in \azj$ is associated with a profit $p_0({\ci})$, for all $1 \leq j \leq n_0$.
Let $S_0$ be a given operation schedule for $I_0$. Consider the input $I$ derived from $I_0$ by adding or removing patients, by changing the possible time intervals for the patients,
or the profits associated with the operations. Suppose that $I$ consists of $n$ patients; each patient $j$ has a set of possible time intervals $\aj$. Each interval ${\ci} \in \aj$ has a profit $p({\ci}) \geq 0$
and a transition cost $\delta({\ci}) \in \mathbb{N}$.
This is either the cost of adding ${\ci}$ to the schedule, if ${\ci} \notin S_0$, or the cost of omitting ${\ci}$ from $S_0$.
 In any feasible schedule $S$ for $I$, at most one interval ${\ci} \in \aj$ is selected, for all $1 \leq j \leq n$,
and the surgery room is occupied by at most one patient at any time $t \geq 0$. The goal is to find a feasible schedule that maximizes
the total profit, while minimizing the aggregate transition cost. In particular, we want to obtain a $(1, \alpha)$-reapproximation
algorithm for the problem, for some $\alpha \in (0,1]$.

\label{sec:related}
\subsection{Main Results}

%Our approximation technique crucially depends
%on the following result.
We prove (in Theorem \ref{lemma:lagrangian_relaxation})
a general result demonstrating the power of Lagrangian relaxation
in solving constrained maximization problems with arbitrary objective functions.

% Using Theorem \ref{lemma:lagrangian_relaxation} we
We use this result to develop
%a general approximation algorithm for
%a fast and simple
a unified approach for solving
subset selection problems with linear constraints.
In particular, given a problem $\Pi$ in this class, and a fixed $\eps \in (0,1)$,
we show that if there exists
an $r$-approximation algorithm for the Lagrangian relaxation of $\Pi$,
for some $r \in (0,1)$, then our technique yields a ratio of
$(\frac{r}{r+1}  -\! \!\eps)$ to the optimal.
We further show (in Section \ref{sec:example}) that this bound is essentially tight, within additive of $\eps$.
Specifically,
there is
%we give an example for
a subset selection problem $\Gamma$
% and show
such that, if
there exists an $r$-approximation algorithm for the Lagrangian relaxation
of $\Gamma$ for some $r \in (0,1)$, there is an input $I$
for which finding the
solutions ${SOL}_1$ and ${SOL}_2$ (for the Lagrangian relaxation)
and combining these solutions
yields at most a ratio of $\frac{r}{r+1}$ to the optimal.

%This shows the tightness of our bound, within additive of $\eps$.
%we give an example of an input $I$ for a subset selection problem
%$\Gamma$, such that the Lagrangian relaxation of $I$ can be solved
%optimally (i.e., $r=1$); applying to this input
%the traditional Lagrangian technique, which
%finds the two solutions ${SOL}_1$, ${SOL}_2$, we cannot obtain from
%${SOL}_1$, ${SOL}_2$ a solution that is better than
%$1/2$-approximation for $I$. This shows the tightness of the bound
%of $r/(r+1)- \eps$ for $r=1$,
%to within additive of $\eps$);
%for $r < 1$ our bound is always better than $r/2$.
%The running time of our algorithms (measured as the number of calls to the
%$r$-approximation algorithm) is {\em linear} in
The number of calls
to the $r$-approximation algorithm, used by our algorithms,
is {\em linear} in
the input size and in $\log (1 / \eps)$, for inputs with cardinality
constraint (i.e., where $w(s)=1$ for all $s \in U$), and polynomial in
the input size and in $\log (1 / \eps)$ for inputs with arbitrary linear
constraint (i.e., arbitrary weights $w(s) \geq 0$).

We apply the technique to obtain efficient approximations for
natural variants of some classic subset selection problems.
In particular, for the budgeted variants of
the real-time scheduling problem we obtain (in Section \ref{sec:bba})
a bound of $(1/3 - \eps)$ for BRS and $(1/4 -\eps)$ for
BCRS.
For budgeted GAP we give (in Section \ref{sec:bgap}) an approximation
ratio of $\frac{1-e^{-1}}{2-e^{-1}} -\eps$.

For BWIS we show (in Section \ref{sec:bwis})
how an approximation algorithm $\mathcal{A}$ for WIS can be used
to obtain an approximation algorithm for BWIS with
the same asymptotic approximation ratio.
%the same asymptotic behavior as $\mathcal{A}$.
More specifically,
let $\mathcal{A}$ be a polynomial time algorithm
that finds in a graph $G$ an independent set whose profit is at least
$f(n)$ of the optimal, where $(i)$ $f(n)=o(1)$ and
$(ii)$ $\log(f(n))$ is polynomial in the input size.\footnote{These two
requirements hold
for most approximation algorithm for the problem.}
Our technique yields an
approximation algorithm which runs in polynomial time and achieves
an approximation ratio of $g(n)=\Theta(f(n))$.
Moreover,
%Furthermore,
%our algorithm has approximation ratio of $g(n)$
%such that
$\lim_{n\rightarrow \infty} \frac{g(n)}{f(n)} =1$.
%(i.e.,
%the approximation ratios of $\mathcal{A}$ and our algorithm have
%the same asymptotic behavior).
Since BWIS generalizes WIS, this implies that the two problems are
essentially equivalent in terms of hardness of approximation.

Our technique can be applied iteratively to obtain an $(\frac{r}{1
+ d r} - \eps)$-approximation algorithm for subset selection
problems with $d > 1$ linear constraints, when there exists an
$r$-approximation algorithm for the non-constrained version of
the problem, for some $r \in (0,1)$.
%(see in \cite{full}).
%AK-fixed ref
%\ref{sec:multi_budgeted}).

It is important to note that the above results, which apply for maximization
problems with {\em linear} constraints, do not
%fully utilize the general
%utilize
%use
exploit the result in Theorem \ref{lemma:lagrangian_relaxation} in
its full generality.
We believe that the theorem will find more uses,
%applications,
e.g., in deriving approximation
algorithms for subset selection problems with {\em non-linear} constraints.

Finally, we show how our technique can be used to develop a general $(1,\alpha)$-reapproximation algorithm for any subset selection problem.
Specifically, given an instance of a reoptimization problem $R(\Pi)$, where $\Pi$ is a subset selection problem with
an $r$-approximation algorithm
$\ca$, we find a $(1,(\frac{r^2}{1+r}-\varepsilon))$-reapproximation algorithm for $R(\Pi)$.
We do so by considering a family of
{\em budgeted reoptimization problems} (see Section \ref{reopt:budgeted}) denoted by $R(\Pi, b)$.
The problem $R(\Pi, b)$  is a restricted version of $R(\Pi)$,
in which we add the constraint that the total transition cost is at most $b$, for some budget $b \geq 0$.
The optimal solution for $\rpim$ is denoted $\co_b$. Note that
$\co_b$ is the profit of the best solution that can be obtained from the initial solution with transition
cost at most $b$. Each of these budgeted sub-problems can be approximated
using the Lagrangian relaxation technique. We search for the lowest budget $b^*$
such that $R(\Pi, b^*)$ has a solution of profit that exceeds a certain threshold.
Using this algorithm, we derive a $(1,\frac{1}{6}-\varepsilon)$-reapproximation algorithm for SRAP (see Section \ref{sec:reopt}).
%This is done by applying as a subroutine an approximation algorithm  for {\em real-time scheduling} \cite{BB00},
%which serves as a {\em base problem}.\footnote{See the details in section \ref{sec:reopt}.}
%%%%%%%%%%%%%%%%%%%%%%%%%%%%%%%%%%%%%%%%%%%%%%%%%%%%%%%%%%%%%%%%%%%%%
%\label{subsubsec:reoptproblems}
%\mysubsubsection{A demonstrate on how to produces a $(1,\alpha)$-reapproximation algorithm}
%We demonstrate the use of the technique described in Section
%\ref{sec:reopt} by presenting reapproximation algorithms for the following problem.

%\paragraph{Technical Contribution}
\subsection{Related Work}
\label{sec:related_work}
%Lagrangian Relaxation:
Most of the approximation techniques based on Lagrangian relaxation are
tailored to handle specific optimization
%minimization/maximization
problems. In solving the {\em k-median} problem through a relation to {\em
facility location}, Jain and Vazirani
%\cite{JV01}
%consider in \cite{JV01} the $k-median problem, and show how
%to solve it
%developed in \cite{JV01} (see also \cite{G96})
developed in \cite{JV01} a general framework for using Lagrangian relaxation to derive
approximation algorithms (see also \cite{G96}).
%This framework is the used in \cite{JV01} to obtain efficient approximation
%for the $k$-median problem, through the relation to facility location.
The framework, that is based on a primal-dual approach, finds initially
two approximate solutions
${SOL}_1$, ${SOL}_2$
%, applies
%and approximation algorithm for the relaxed problem. Initially, two
%$s_1, s_2$
for the Lagrangian relaxations
$\Pi(\lambda_1)$, $\Pi(\lambda_2)$ of a problem $\Pi$, for carefully selected values of
$\lambda_1, \lambda_2$; a
%the idea is to find a
{\em convex combination} of these solutions yields a (fractional)
%approximate
solution which uses the budget $L$.
This solution is then rounded to obtain an integral
solution that is a good approximation for the original problem.
Our approximation technique (in Section \ref{sec:technique}) differs
from
%We develop (in Section \ref{sec:technique})
%an approximation technique that simplifies
the technique of \cite{JV01} in
two ways. First,
%as
it does not require
rounding a fractional solution: in fact, we do not attempt to
combine the
solutions
${SOL}_1$, ${SOL}_2$,
%$s_1, s_2$,
but rather, examine each separately and compare the
two feasible solutions which can be easily derived
from
${SOL}_1$, ${SOL}_2$,
%$s_1, s_2$,
using
an efficient transformation of the non-feasible solution,
${SOL}_2$,
%$s_2$,
to a
feasible one.
%(see in Section \ref{sec:technique}).
%a {\em packing} procedure.
%More important,
Secondly, the framework of \cite{JV01}
crucially depends on a primal-dual interpretation of the approximation
algorithm for the relaxed problem, which is not required here.

%Recently,
K\"{o}nemann et al. considered in \cite{KPS06} a technique for
solving general partial cover problems. The technique
%, which
builds on the framework of \cite{JV01}, namely, an instance of a problem
in this class is
solved by initially finding the two solutions
${SOL}_1$, ${SOL}_2$
%$s_1, s_2$,
and generating a
solution
${SOL}$,
which combines these two solutions.
A comprehensive survey of other work is given in \cite{M07}.\footnote{For
conditions under which Lagrangian relaxation can be used to
solve discrete/continuous optimization problems see, e.g., \cite{NW99}.}

There has been some earlier work on
%Several earlier studies consider on
using Lagrangian relaxation to solve subset selection problems.
% has been studied in several papers.
%Naor el al.
The paper \cite{NST07} considered a subclass of the class
%of instances
of the subset selection problems that we study here.\footnote{This subclass includes the {\em
real-time scheduling} problem.}
 Using the framework of \cite{JV01},
the paper claims
to obtain an approximation ratio of $r -\eps$ for any
problem in this subclass,
%AK- text ommited from the footnote
%problem, a special case of bandwidth allocation
%where the resource requirement of any instance in some activity
%$A_i$, $1 \leq i \leq m$, is equal to $1$.
given a $r$-approximation algorithm for the
Lagrangian relaxation of the problem (satisfying certain properties).
Unfortunately, this approximation ratio was shown to be incorrect \cite{S07}. Berget et al. considered
in \cite{BBGS08} the budgeted matching problem and the budgeted matroid
intersection problem.
%Using the framework of \cite{JV01}, the authors derive the
The paper gives the
first polynomial time approximation schemes for these problems.
%AK- text added
The schemes, which are based on Lagrangian relaxation, merge the two obtained
solutions using some strong combinatorial properties of the problems.

The non-constrained variants of the subset selection problems that we study
here are well studied.
For known results on real-time scheduling
%AK- text modified, citations added
and related problems see, e.g., \cite{BB00,COR06,BGNS01,BBCR08}.
%For approximation results as well as classification of the APX-hard special
%cases of GAP, see, e.g.,
%Comprehensive
Surveys of known results for the generalized assignment problem
are given, e.g., in \cite{CK06,CKR06,FV06,FG06}.

Numerous approximation algorithms have been proposed and analyzed for the
maximum (weight) independent set problem.
Alon at al. \cite{AFWZ95} showed that IS
%maximum weight independent set
cannot be approximated within factor $n^{-\eps}$ in polynomial time,
where $n=|V|$ and $\eps > 0$ is some constant, unless $P=NP$.
The best known approximation ratio of $\Omega(\frac{\log^2 n}{n})$ for WIS
on general graphs is due to Halld{\'o}rsson \cite{Ha00}.
%polynomial time, which is
%, for the best of our knowledge,
%the best known approximation ratio for this problem.
A survey of other known results for IS and WIS
can be found e.g., in \cite{H98,H00}.
%The paper \cite{hpp01}
%considers the special case of B\_MWIS where all vertices have unit weights
%and presents
%an optimal algorithm for the problem on trapezoid
%graphs.\footnote{This solves also the problem with unit weights
%on interval graphs, a subclass of trapezoid graphs.}

To the best of our knowledge, approximation algorithms
for the budgeted variants of the above problems
are given here for the first time.
%In  show an improvement of
Our bound for BGAP (in Theorem \ref{thm:bgap}) was improved in \cite{ariel_thesis} to $1 - 1/e$.

There is a wide literature on scenarios leading to reoptimization problems, however, most of the earlier studies refer
to a model in which the goal is to find an optimal solution for a modified problem instance, with no transition costs incurred in the process
(see \cite{STT12} and the references therein). In this paper, we adopt the reoptimization model introduced in \cite{STT12}.
In this model, it is shown in \cite{STT12} that for any subset selection problem $\Pi$ that is polynomially solvable, there is
a polynomial time $(1,1)$-reoptimization algorithm for $R(\Pi)$.
%We use in this paper the reoptimization model introduced in \cite{STT12}
%\item
%The problems studied in this paper: Max cover, MAX SAT,
%Throughput maximization and Max IS (Mention the algorithms that we use for
%solving the relaxation).
%\end{itemize}
%\subsection{Constrained Subset Selection Problems}
%%%%%%%%%%%%%%%%%%%%%%%%%%%%%%%%%%%%%%%%%%%%%%%%%%%%%%%%%%%%%%%%%

\section{Lagrangian Relaxation Technique}
\label{sec:technique}

Given a universe $U$, let
$f:U \rightarrow \mathbb{N}$ be some objective function,
and let
%$w_s \in \mathbb{N}$ be the weight associated with any element $s \in U$.
$w:U \rightarrow \mathbb{R^+}$ be a non-negative weight function.
%$w:D \rightarrow \mathbb{N}$.
Consider the problem $\Pi$ of maximizing $f$ subject to
a budget constraint $L$ for $w$, as given in \eqref{prb:general_constraint},
and the Lagrangian relaxation of $\Pi$, as given in
\eqref{prb:general_lagrange}.
%\begin{eqnarray}
%\label{prb:general_constraint}
%\Pi =&  {\max_{s \in D} }  & f(s) \\
%\nonumber
%&\mbox{subject to:} & w(S) \leq L
%\end{eqnarray}
%
%Consider the \emph{Lagrangian relaxation} of $\Pi$, given by
%\begin{equation}
%%\[
%\label{eq:general_lagrange}
%\Pi (\lambda) :~ \max_{s\in U}  {f(s) - \lambda\cdot w(s).}
%%\]
%\end{equation}

We assume that the value of an optimal solution $s^*$ for $\Pi$
satisfies $f(s^*) \geq 1$.
%is non-negative.
For some $\eps' >0$, suppose that
\begin{equation}
\label{eq:lambda_one_two}
\lambda_2 \leq \lambda_1 \leq \lambda_2 + \eps'.
\end{equation}
The heart of our approximation technique is the next result.
%\\
%{\bf Theorem 1:} {\em
\begin{theorem}
\label{lemma:lagrangian_relaxation}
For any $\eps > 0$ and $\lambda_1, \lambda_2$ that satisfy
(\ref{eq:lambda_one_two}) with $\eps'=\eps/L$, let $s_1={SOL_1}$ and
$s_2={SOL_2}$ be
%Let $\eps'=\frac{\eps}{L}$.If $s_1,s_2$ are
$r$-approximate solutions for
 $\Pi(\lambda_1),\Pi(\lambda_2)$,
such that
$w(s_1) \leq L \leq w(s_2)$. Then
 %$w(s_1) \leq L \leq w(s_2)$, then
for any $\alpha  \in [1- r, 1]$, at least one of the following holds:
 \begin{enumerate}
 \item
 $f(s_1) \geq \alpha r f(s^*)$
 \item
 \label{s_2_case}
 $f(s_2)
%was \geq
(1-\alpha -\eps)f(s^*) \frac{w(s_2)}{L} $.
\end{enumerate}
\end{theorem}
%{\em  where $S^*$ is an optimal solution for} $\Pi$.
\begin{proof}
Let $L_i=w(s_i)$, $i=1,2$, and $L^* = w (s^*)$.
From \eqref{prb:general_lagrange} we have that

%Since
%$f(s_i)$ is a $r$-approximation for $\Pi(\lambda_i)$, we
%get:
%\[
%f(S_i)-\lambda_i \cdot L_i \geq r (f(S^*)-\lambda_iL^*),
%\]
%i.e.,
\begin{equation}
\label{lagrange_prop1}
f(s_i)- r f(s^*) \geq \lambda_i (L_i-r L^*).
\end{equation}
Assume that, for some
$\alpha \in [1- r, 1]$,
%$1 - r \leq \alpha \leq 1$,
it holds that
$f(s_1) < \alpha r f(s^*)$,
%then taking $i=1$ in \eqref{lagrange_prop1} we have
then
\[
(\alpha -1)r f(s^*) > f(s_1)-r f(s^*) \geq
\lambda_1(L_1- r L^*) \geq -r \lambda_1 L^* \geq
-r \lambda_1 L.
\]
The second inequality follows from \eqref{lagrange_prop1}, the third
inequality from the fact that $\lambda_1 L_1 \geq 0$, and the last inequality
holds due to the fact that $L^* \leq L$.
%And since $\lambda_1 < \lambda_2 +\eps'$ we get:
Using \eqref{eq:lambda_one_two}, we have
\begin{equation}
\label{eq:lambda_and_f}
%\[
\frac{ (1- \alpha ) f(s^*)} {L} <
   \lambda_1 < \lambda_2 +\eps'.
%\]
\end{equation}
%using \eqref{lagrange_prop1} with $i=1$ we get
Since $\eps'= \eps/L$, we get that
\begin{eqnarray*}
f(s_2)  &\geq& \lambda_2(L_2-L^*)+r f(s^*)
>
\left( \frac{(1-\alpha) f(s^*) }{L} - \eps' \right) (L_2-L) +r f(s^*)
\\
&\geq&
(1-\alpha)  f(s^*) \frac{L_2}{L} - \eps' L_2
\geq
(1- \alpha -\eps' L)\frac{L_2}{L}f(s^*)
= (1-\alpha -\eps)\frac{L_2}{L}f(s^*)
\end{eqnarray*}
The first inequality follows from \eqref{lagrange_prop1}, by taking $i=2$,
and the second inequality is due to \eqref{eq:lambda_and_f} and the fact that
$L^* \leq L$.
The third inequality holds since $r \geq 1- \alpha$, and the last inequality
follows from the fact that $f(s^*) \geq 1$.
%\hspace*{\fill} $\Box$ \vskip \belowdisplayskip
\end{proof}

We summarize the above discussion in the next theorem,
which is the heart of our technique:

\begin{theorem}
\label{lemma:lagrangian_relaxation}
Let $\eps'=\frac{\eps}{L}$.If $S_1,S_2$ are
 $r$-approximation for
$\Pi(\lambda_1),\Pi(\lambda_2)$,
$\lambda_2 \leq \lambda_1 < \lambda_2 +\eps'$,
and $w(S_1) \leq L \leq w(S_2)$, then
for all $1- r \leq \alpha  \leq  1 $, one of the
following holds:
\begin{enumerate}
\item
$f(S_1) \geq \alpha r f(S^*)$
\item
\label{s_2_case}
$f(S_2) \geq (1-\alpha -\eps)f(S^*) \frac{L_2}{L} $,
\end{enumerate}
where $S^*$ is an optimal solution for $\Pi$.
\end{theorem}

Theorem \ref{lemma:lagrangian_relaxation} asserts
that at least one of the solutions $s_1,s_2$ is
{\em good} in solving our original problem, $\Pi$.
If $s_1$ is a good solution then
we have an $\alpha r$-approximation for $\Pi$,
otherwise we need to find a way to convert $s_2$ to a
solution $s'$ such that $w(s') \leq L$ and $f(s')$ is
a good approximation for $\Pi$.
Such conversions
%(called below \emph {roundings})
are presented in Section \ref{sec:rounding} for
a class of {\em subset selection problems
with linear constraints}.
Next, we show how to find two solutions which satisfy
the conditions of Theorem \ref{lemma:lagrangian_relaxation}.

\subsection{Finding the Solutions $s_1, s_2$}
\label{sec:finding}
Suppose that we have an algorithm $\mathcal{A}$ which
finds a $r$-approximation for $\Pi(\lambda)$, for
any $\lambda \geq 0$. Given an input $I$ for $\Pi$,
denote the solution which
$\mathcal{A}$ returns for $\Pi(\lambda)$ by
$\mathcal{A}(\lambda)$,
and assume that
%in solving $\Pi$ optimally for the input $I$
it is sufficient to consider $\Pi(\lambda)$ for
$\lambda \in (0, \lambda_{max})$,
where $\lambda_{max} = \lambda_{max}(I) $ and
$w(\mathcal{A}(\lambda_{max}))\leq L$.

Note that if $w(\mathcal{A}(0)) \leq L$ then
$\mathcal{A}(0)$ is a $r$-approximation
for $\Pi$; otherwise, there exist $\lambda_1,
\lambda_2 \in (0,\lambda_{max})$ such that
$\lambda_1,\lambda_2$, and $s_1=\mathcal{A}(\lambda_1),
s_2=\mathcal{A}(\lambda_2)$ satisfy \eqref{eq:lambda_one_two} and
the conditions of
Theorem \ref{lemma:lagrangian_relaxation}, and $\lambda_1, \lambda_2$
can be easily found using binary search.
Each iteration of the binary search requires a single
execution of $\mathcal{A}$ and reduces the size of the
search range by half. Therefore, after
$R = \ceil{ \log (\lambda_{max})+
\log(L) + \log(\eps ^{-1})}$
iterations, we have two solutions which satisfy
the conditions of the theorem.

Note that the values of $\lambda$ used during
the execution of the algorithm can be represented by
$O (\log(\lambda_{max})+\log(L)+\log(\eps^{-1}))$
bits; thus, we keep the
problem size polynomial in its original
size and in $\log(\eps^{-1})$.

\begin{theorem}
\label{lemma:finding}
Given an algorithm $\mathcal{A}$ which outputs
an $r$-approximation for $\Pi(\lambda)$, and
$\lambda_{max}$, such that
$w(\mathcal{A}(\lambda_{max})) \leq L$,
an $r$-approximate solution or two solutions $s_1, s_2$
which satisfy the conditions of Theorem
\ref{lemma:lagrangian_relaxation} can be found by using binary
search.
%a binary search procedure can be used to find
%solutions $s_1, s_2$ which satisfy the conditions of
%Theorem \ref{lemma:lagrangian_relaxation}, or a
%$r$-approximation for $\Pi$.
This requires $\ceil{ \log (\lambda_{max})+
\log(L) + \log(\eps ^{-1})}$
executions of $\mathcal{A}$.
\end{theorem}

We note that when $\mathcal{A}$ is a randomized approximation
algorithm whose
{\em expected} performance ratio is $r$,
a simple binary search may not output solutions that
satisfy the conditions of Theorem \ref{lemma:lagrangian_relaxation}.
In this case,
%for some $\bar{\eps}>0$,
we repeat the executions
of $\mathcal{A}$ for the same input
and select the solution of maximal value.
For some pre-selected values $\beta >0$ and $\delta >0$,
we can guarantee that the probability that any of the used solutions
is not a $(r-\beta)$-approximation is
bounded by $\delta$.  Thus, with appropriate selection of the values of
$\beta$ and $\delta$, we get a result similar to the result in
Theorem \ref{lemma:lagrangian_relaxation}.
We discuss this case in detail in the full version of the paper.

\section{Approximation Algorithms for Subset Selection Problems}
\label{sec:rounding}
In this section we develop an approximation technique for subset selection
problems. We start with some definitions and notation.
Given a universe $U$, let $X \subseteq 2^U$ be a domain, and
$f:X \rightarrow \mathbb{N}$ a set function. For a subset $S \subseteq U$, let
$w(S)= \sum_{s \in S} w_s$, where $w_s \geq 0$ is the \emph{weight} of the element $s \in
U$.

\begin{definition}
The problem
%\begin{eqnarray*}
\[
\Gamma:~ \max_{S \in X} ~f(S)~~~
\]
\negA
\negA
%\nonumber
\begin{equation}
\label{eq:subset_constraint}
\mbox{subject to:} ~ w(S) \leq L
\end{equation}
%\end{eqnarray}
is a \emph{subset selection problem with a linear constraint} if $X$
is a lower ideal, namely, if $S \in X$ and
$S' \subseteq S$ then $S' \in X$, and $f$ is
%AK - added $f(\emptyset}=0$
a linear non-decreasing set function with $f(\emptyset)=0$.\footnote{For simplicity,
we assume throughout the discussion that $f(\cdot)$ is a {\em linear}
function; however, all of the results in this section hold also for
the more general case where $f: 2 ^S \rightarrow \mathbb{N}$ is a
non-decreasing {\em submodular}
set function, for any $S \in X$.}
%Our results
%also hold in the case where for every $S \in X$,
%when restricting $f$ to the subsets of $S$, $f$ is a non-decreasing submodular
%set function. This was omitted for the simplicity of the paper } with $f(\emptyset)=0$.
\end{definition}

Note that subset selection problems with linear constraints
are in the form of \eqref{prb:general_constraint}, and
the Lagrangian relaxation of any problem $\Gamma$ in this class
is $\Gamma(\lambda)=\max_{S\in X} f(S)- \lambda w(S)$; therefore,
the results of Section \ref{sec:technique} hold.

Thus, for example, BGAP  can be formulated
as the following subset selection problem with linear constraint. The
universe $U$
consists of all pairs $(i,j)$ of item $1 \leq i \leq n$
and bin $ 1 \leq j \leq m$. The domain $X$ consists of
all the subsets $S$ of $U$, such that each item appears at
most once (i.e., for any item $1 \leq i \leq n$, $\left| \{ (i',j') \in S : i'=i \}\right| \leq 1$),
and the collection of items that appears with a bin $j$, i.e.,
$\{ i : (i,j) \in S \}$ defines
a feasible assignment of items to bin $j$. It is easy to see that $X$
is indeed a lower ideal.
The function $f$ is
$f(S)= \sum_{(i,j)\in S} f_{i,j}$, where $f_{i,j}$ is the profit
from the assignment of item $i$ to bin $j$, and $w(S) = \sum_{(i,j)\in S} w_{i,j}$
where $w_{i,j}$ is the size of item $i$ when assigned to bin $j$.

The Lagrangian relaxation of BGAP is then
$$\max_{S \in X} f(S) -\lambda w(S)=  \max_{S \in X} \sum_{(i,j)\in S}
(f_{i,j} - \lambda w_{i,j})\enspace{.}$$
The latter can be interpreted as the following
instance of GAP: if $f_{i,j}-\lambda w_{i,j} \geq 0$ then
set $f_{i,j}-\lambda w_{i,j}$ to be the profit from assigning item $i$ to bin $j$; otherwise,
make item $i$ infeasible for bin $j$ (set the size of item $i$
to be greater than the capacity of bin $j$).

%Given a $r$-approximation algorithm $\ca$ for $\Gamma(\lambda)$,
We now show how the Lagrangian relaxation technique described in Section
\ref{sec:technique} can be applied to subset selection problems.
Given a problem $\Gamma$ in this class, suppose that $\ca$
is a $r$-approximation algorithm for $\Gamma(\lambda)$, for some $r \in
(0,1)$. To find $\lambda_1, \lambda_2$ and ${SOL}_1, {SOL}_2$,
the binary search of Section \ref{sec:finding}
can be applied over the range $[0,p_{max}]$, where
\begin{equation}
\label{eq:def_pmax}
p_{max} = \max_{s \in U} f(s)
\end{equation}
is the maximum profit of any element in the universe $U$.
To obtain the solutions $S_1, S_2$ which correspond to $\lambda_1,
\lambda_2$, the number of calls to ${\ca}$ in the binary search is
bounded by $O(\log(\frac{L \cdot p_{max}}{\eps}))$.

Given the solutions $S_1,S_2$
satisfying the conditions of
Theorem \ref{lemma:lagrangian_relaxation},
consider the case where, for some $\alpha \in [1- r, 1]$,
property \ref{s_2_case} (in the theorem)
holds. Denote the value of an optimal solution for $\Gamma$
by $\mO$. Given a solution $S_2$ such that
\begin{equation}
\label{eq:beta_appx}
f(S_2) \geq (1-\alpha - \eps) \frac{w(S_2)}{L} \cdot \mO,
\end{equation}
our goal is to find a solution $S'$ such that
$w(S') \leq L$ (i.e., $S'$ is valid for $\Gamma$),
and $f(S')$ is an approximation for $\mO$.
We show below how $S'$ can be obtained from $S_2$. We first consider
(in Section \ref{sec:rounding_card}) instances with unit weights. We
then describe (in Section \ref{sec:rounding_linear}) a scheme for
general weights. Finally, we give (in Section \ref{sec:enumeration})
a scheme which
yields improved approximation ratio for general instances, by
applying enumeration.
%In the following subsections we show three
% conversion scheme to obtain $S'$.
%One for the case of unit weight (e.g. $W(S)=|S|$), one for general
%weight that does not use enumeration, and one that uses
%enumeration to obtain a better approximation ratio
%for the general weight solution.

\subsection{Unit Weights}
\label{sec:rounding_card}
Consider first the special case where
$w_s = 1$ for any $s\in U$ (i.e., $w(S)= |S|$; we refer to
(\ref{eq:subset_constraint}) in this case
as \emph{cardinality constraint}).
%%%Moved
\comment{
To obtain $S'$, select the $L$ elements in $S_2$ with
the highest profits.\footnote{When $f$ is a submodular
function,
iteratively select the element $s \in S_2$ which maximizes  $f(T\cup \{s\})$,
where $T$ is the subset of elements chosen in
the previous iterations.}
It follows from (\ref{eq:beta_appx}) that
$f(S') \geq (1-\alpha - \eps) \cdot \mO = (\frac{r}{\1+r} - \eps)\mO$.
}

Suppose that we have solutions $S_1,S_2$
which satisfy the conditions
of Theorem \ref{lemma:lagrangian_relaxation}, then by
taking $\alpha = \frac{1}{1+r}$ we get that
either $f(S_1) \geq (\frac{r}{1+r} -\eps) \mO$,
or $f(S_2) \geq ( \frac{r}{1+r} -\eps)
\frac{w(S_2)}{L} \mO$.
If the former holds
%$f(S_1) \geq (\frac{r}{1+r} -\eps) \mO$,
then we
have a ($\frac{r}{1+r}-\eps$)-approximation
for the optimum; otherwise, $f(S_2) \geq ( \frac{r}{1+r} -\eps)
\frac{w(S_2)}{L} \mO$.
To obtain $S'$, select the $L$ elements in $S_2$ with
the highest profits.\footnote{When $f$ is a submodular
function,
iteratively select the element $s \in S_2$ which maximizes  $f(T\cup \{s\})$,
where $T$ is the subset of elements chosen in
the previous iterations.}
It follows from (\ref{eq:beta_appx}) that
$f(S') \geq (1-\alpha - \eps) \cdot \mO = (\frac{r}{1+r} - \eps)\mO$.
Combining the above with the result of Theorem
\ref{lemma:finding}, we get the following.
\begin{theorem}
\label{lemma:cardinality}
Given a subset selection problem $\Gamma$ with unit weights,
an algorithm $\mathcal{A}$ which
yields a $r$-approximation for $\Gamma(\lambda)$ and
$\lambda_{max}$, such that
$w(\mathcal{A}(\lambda_{max})) \leq L$, a
$(\frac{r}{r+1} -\eps)$-approximation for
$\Gamma$ can be derived by using $\mathcal{A}$ and selecting among
$S_1, S'$ the set with highest profit.
The number of calls to $\mathcal{A}$ is
$O(\log(\frac{L \cdot p_{max}}{\eps}))$, where $p_{max}$ is given in (\ref{eq:def_pmax}).
\end{theorem}

%hence, when  $S'$ is obtained from
%$S_2$ by the process mentioned above, we have
%$f(S') \geq (\frac{r}{1+r} -\eps) \mO$. This means,
%that selecting the set among $S_1,S'$ with the highest
%profit always gives a ($\frac{r}{1+r}-\eps$)-approximation
%for the optimum of $\Gamma$.

\comment{
Combining the above with the result of Theorem
\ref{lemma:finding}, we get the following.
\begin{theorem}
\label{lemma:cardinality}
Given a subset selection problem $\Gamma$ with unit weights,
an algorithm $\mathcal{A}$ which
yields a $r$-approximation for $\Gamma(\lambda)$, and
$\lambda_{max}$, such
$w(\mathcal{A}(\lambda_{max})) \leq L$, a
$(\frac{r}{r+1} -\eps)$-approximation for
$\Gamma$ can be derived by using $\mathcal{A}$.
The number of calls for $\mathcal{A}$ is
%bounded by a polynomial
linear in the input size and $(\log (1/\eps))$,
plus a linear time processing.
\end{theorem}
}
%We note that our approximation ratio of $\frac{r}{r+1}$
%maintains the approximation ratio of $r$ for the relaxed version of the
%problem, to within factor of $2$. In fact, as $r$ decreases, we get that
%our approximation ratio approaches $r$.
%The ratio $(\frac{r}{1+r})/r$ is always grater than $0.5$,
%and approaches $1$ as $r$ decreases. For example, for $r= 0.5$,
%$\frac{r}{1+r}=\frac{1}{3}$, and this ratio is $\frac{2}{3}$.

\subsection{Arbitrary Weights}
%with Linear Constraint}
\label{sec:rounding_linear}

For general element weights, we may assume w.l.o.g. that,
for any $s \in U$, $w_s \leq L$.
%(since ise the element can be removed from the original problem as it
%would not be in any valid solution).
We partition $S_2$ to a collection of
%AK- added the bound for the number of disjoint set, it is required for
% the algorithm to work
up to $\frac{2W(S_2)}{L}$
disjoint sets $T_1, T_2, \ldots$
%$T_1,...,T_k$,
such that $w(T_i) \leq  L$ for all $i \geq 1$. A simple
way to obtain such sets is by adding elements of $S_2$ in arbitrary
order to $T_i$ as long as  we do not exceed the budget $L$.
A slightly more efficient implementation has a running time that
is linear in the size of $S_2$ (details omitted).
%The above algorithm can be implemented in linear time.
%The next result is easy to verify (details omitted).
%\begin{observation}
%\label{observation:rounding_comb}
%Let $N$ be the number of sets $T_i$ output by the algorithm, then
%$N \leq  \frac{2w(S_2)}{L}$.
%\end{observation}

%[PROOF CAN BE MOVED TO THE APPENDIX, ITS NOT HARD]

%AK- the observation is part of the algorithm (otherwise the algorithm is not well defined)
%\begin{observation}
%\label{observation:rounding_comb}
%Let $N$ be the number of sets $T_i$ output by the algorithm, then
%$N \leq  \frac{2w(S_2)}{L}$.
%\end{observation}

\begin{lemma}
Suppose that $S_2$ satisfies \eqref{eq:beta_appx} for some $\alpha \in [1-
r,1]$, then there exists $i \geq 1$ such
that $f(T_i) \geq \frac{1-\alpha - \eps}{2}\cdot \mO$.
\end{lemma}
\begin{proof}
%AK - proof modified
Clearly, $f(T_1)+...+f(T_N) = f(S_2)$,
where $N \leq \frac{2w(S_2)}{L} $ is the number of disjoint
sets. By the pigeon hole principle
there exists $1 \leq i \leq N$ such that
$f(T_i) \geq \frac{f(S_2)}{N} \geq
\frac{L \cdot f(S_2)}{2 w( S_2)} \geq
\frac{1 - \alpha -\eps}{2} \cdot \mO$.
%\hspace*{\fill} $\Box$ \vskip \belowdisplayskip
\end{proof}

%\begin{claim}
%There exists $i$ such that $f(T_i) \geq \frac{\beta}{2}\cdot \mO$.
%\end{claim}
\comment{
Trivial for the linear case, hence omitted
\begin{proof}
As before, let $f(S)=f(c_2,S)$.
%Let $g=f_{C_2}$ as in section \ref{sec:rounding_card}
%non-decreasing submodular function.
Then
$f(T_1)+...+f(T_N) \geq f(S_2)$, and by the
pigeon hole principle and Observation \ref{observation:rounding_comb}
there is $i$ such that
$f(c_2,T_i) \geq \frac{f(S_2)}{N} \geq
\frac{L \cdot f(S_2)}{2 w( S_2)} \geq
\frac{\beta}{2} \cdot \mO$.
\hspace*{\fill} $\Box$ \vskip \belowdisplayskip
\end{proof}
}

Assuming we have solutions $S_1, S_2$ which satisfy the conditions
of Theorem \ref{lemma:lagrangian_relaxation}, by
taking $\alpha = \frac{1}{1+2r}$ we get that
either $f(S_1) \geq (\frac{r}{1+2r} -\eps) \mO$,
or $f(S_2) \geq ( \frac{2r}{1+2r} -\eps)
\frac{w(S_2)}{L} \mO$ and can be converted to $S'$ (by setting
$S'=T_i$ for $T_i$ which maximizes $f(T_i)$),
such that $f(S') \geq (\frac{r}{1+2r} -\eps) \mO$, i.e.,
we get a $(\frac{r}{1+2r}- \eps)$-approximation for $\Gamma$.

Combining the above with the result of Theorem
\ref{lemma:finding}, we get the following.
\begin{theorem}
\label{lemma:general_case}
Given a subset selection problem  $\Gamma$ with a linear
constraint, an algorithm $\mathcal{A}$ that
yields an $r$-approximation for $\Gamma(\lambda)$, and
$\lambda_{max}$, such that
$w(\mathcal{A}(\lambda_{max})) \leq L$, an
$(\frac{r}{2r+1} -\eps)$-approximation for
$\Gamma$ can be obtained using
$O(\log(\frac{L \cdot p_{max}}{\eps}))$ calls to $\mathcal{A}$, where $p_{max}$ is given in (\ref{eq:def_pmax}).
\end{theorem}

\subsection{Improving the Bounds via Enumeration}
\label{sec:enumeration}
In this section we present an algorithm that
uses enumeration to obtain a new problem, for which
we apply our Lagrangian relaxation technique. This enables
to improve the approximation ratio in
Section \ref{sec:rounding_linear}
to match the bound obtained for unit weight inputs (in Section
\ref{sec:rounding_card}).\footnote{The running time
%applicability
when applying enumeration depends on the size of the universe (which may be
super-polynomial in the input size; we elaborate on that in
Section \ref{sec:bba}).}
\comment{
Though enumeration is applicable in
many cases, for
some cases the interpretation of a problem as a subset selection
problem yields an unbounded or super-polynomial number of elements,
as in the case of budgeted bandwidth allocation, when the
input is given with time windows (we elaborate on that in
Section \ref{sec:bba}). In such cases enumeration,
of even a small number of elements, cannot be implemented in
polynomial time.
}

For some $k \geq 1$, our algorithm initially `guesses' a subset $T$
of (at most) $k$ elements with the highest profits in some optimal solution.
Then, an approximate solution is obtained by adding
elements in $U$, whose values are bounded by $f(T)/|T|$.
Given a subset $T \subseteq U$, we define
$\Gamma_T$, which can be viewed as the sub-problem that `remains' from $\Gamma$
once we select $T$ to be the initial solution.
Thus, we refer to $\Gamma_T$ below as the {\em residual problem with
respect to T}.
Let
\begin{equation}
\label{eq:largest_density}
X_T = \left\{ S \left | ~ S \cap T =\emptyset, S \cup T \in X,
\mbox{ and } \forall s \in S: f(\{s\}) \leq \frac{f(T)}{|T| }\right. \right\}
\end{equation}
Consider the residual problem  $\Gamma_T$ and its Lagrangian relaxation
$\Gamma_T(\lambda)$:

\[
\begin{array}{lcr}
\begin{array}{ll}
\Gamma_T & \mbox{maximize~} f(S)
~~~~~~~
\\
 \mbox{subject~to}:&   S \in X_T\\
 & w(S) \leq L-w(T)
\end{array}     &
 ~~
 &
 \begin{array}{ll}
 \Gamma_T(\lambda)&\mbox{maximize~} f (S) -\lambda w(S)\\
 \mbox{subject~to}:& S \in X_T
 \end{array}
\end{array}
\]

In all of our examples, the residual problem $\Gamma_T$ is a smaller instance of
the problem $\Gamma$, and therefore, its Lagrangian relaxation is an instance
of the Lagrangian relaxation of the original problem. Assume that we have an approximation algorithm
$\mathcal{A}$ that, given $\lambda$ and a pre-selected set $T \subseteq U$ of
at most
$k$ elements, for some constant $k >1$, returns an $r$-approximation for
$\Gamma_T(\lambda)$ in polynomial time (if there is
a feasible solution for $\Gamma_T$). Consider the following
algorithm, in which we take $k=2$:
%\\
%\\
\begin{center}
\fbox{
\begin{minipage}{0.8\textwidth}
\begin{alg} General approximation algorithm
\label{alg:general}
\begin{enumerate}
 \item
 \label{enum:main_loop}
 For any $T \subseteq U$ such that $|T| \leq k$,
 find solutions $S_1,S_2$ (for $\Gamma_T(\lambda_1),\Gamma_T(\lambda_2)$ respectively)
  satisfying the conditions of
 %parameters for $\Gamma_T$ as in
 Theorem \ref{lemma:lagrangian_relaxation} with respect to the problem $\Gamma_T$. Evaluate
 the following solutions:

 \begin{enumerate}
 \item
 $T \cup S_1$
 \item
 \label{enum:s2_bags}
 Let $S'= \emptyset$, add elements to $S'$ in the
 following manner:

 Find an element $x \in S_2 \backslash S'$ which maximizes
 the ratio $\frac{f(\{x\})}{w_x}$.
 If $w(S' \cup \{x\}) \leq L-w(T)$ then add $x$ to $S'$
 and repeat the process, otherwise return $S' \cup T$ as a
 solution.
 \end{enumerate}
 \item
 Return the best of the solutions found in Step 1.
 \\
 \end{enumerate}
 \end{alg}
 \end{minipage}
 }
 \end{center}
%\vspace{0.5cm}
 Let $\mO= f(S^*)$ be an optimal solution for
 $\Gamma$, where $S^*= \{x_1, \ldots , x_h \}$.
Order the elements in $S^*$ such that
$f(\{x_1\}) \geq f(\{x_2\})\geq \ldots \geq f(\{x_h\})$.
\comment{
REMOVED DUE TO SUBMODULARITY
$f(\{x_1\}) \geq f_{\{x_1 \}} (\{x_2\}) \geq \cdot \! \cdot \! \cdot
f_{\{ x_1, x_2, \ldots , x_{i-1} \}} (\{ x_i \}) \geq
\cdot \! \cdot \! \cdot$,
 namely,
for any $1 \leq i \leq h$, $x_i$ is the element for which
$f_{\{ x_1, x_2, \ldots , x_{i-1} \}} (\{ x_i \})$ is maximized.
}
\begin{lemma}
%\label{lemma:decreasing_average_density}
Let $T_i = \{x_1, \ldots, x_i \}$, for some $1 < i \leq h$,
then for any $j > i$, $f(\{x_j \}) \leq \frac{f(T_i)}{i}$.
\end{lemma}
%We give the proof in the full version of the paper.)

In analyzing our algorithm, we consider the iteration in which $T=T_k$.
Then $S^* \setminus T_k$ is an optimal solution for $\Gamma_T$
%We get that in solving
%optimally $\Gamma_T$ we consider the solution $S^* \setminus T_k$
(since $S^*
\setminus T_k \in X_{T_k}$ as in \eqref{eq:largest_density}); thus, the optimal value
for $\Gamma_{T_k}$ is at least
$f( S^* \backslash T_k) = f(S^*) - f(T_k)$.

\comment{
 by decreasing
 profits $x_1,x_2,...,x_n$  - $x_1$ is the element
 for which $f(x)$ is maximal, $x_2$ is the element
 for which $f_{ \{x_1\}} (x_2)$ is maximal and so
 forth. Note that for any $i$, let
 $T_i= \{x_1, ... ,x_i\}$ and every $j> i$  then
 $f_{T_i} (x_j) \leq \frac{f(T_i)}{i}$. Which mean
 that the optimal value for $\Gamma_{T_k}$ is
 at least $f_{T_k}( S^* \backslash T_k) = f(S^*) - f(T)$.
}

\begin{lemma}
\label{lemma:enum_s'}
Let $S'$ be the set generated from $S_2$ by the
process in Step 1(b) of the algorithm.
%\ref{enum:s2_bags}
Then
$f(S') \geq f(S_2) \frac{L-w(T)}{w(S_2)}
             - \frac{f(T)}{|T|}$
\end{lemma}

\begin{proof}
Note that the process cannot terminate when $S'=S_2$
since $w(S_2) > L - w(T)$. Consider the first element
$x$ that maximized the ratio
$\frac{f(\{x\})}{w_x}$,
but was not added to $S'$, since
$w(S' \cup \{x\}) > L- w(T)$.
By the linearity of $f$,
it is clear that
\begin{enumerate}
\item[(i)]
$\frac{f(S' \cup \{x\} ) } {w(S' \cup \{x\})}
\geq \frac{f( \{x\})}{w_x}$, and
 \item[(ii)]
For any $y \in S_2 \backslash (S' \cup \{x\})$,
$\frac{f(\{y\})} {w_y}  \leq \frac{f(\{x\})}{w_x}$.
\comment{
\begin{enumerate}
\item[(i)]
$\frac{f(S' \cup \{x\} ) } {w(S' \cup \{x\})}
\geq \frac{f( \{x\})}{w_x}$
\item[(ii)]
For any $y \in S_2 \backslash (S' \cup \{x\})~~$
$\frac{f(\{y\} )}
      {w_y}
  \leq \frac{f(\{x\})}{w_x}$
\end{enumerate}
}
\end{enumerate}
Thus, we get that for any
$y \in S_2 \backslash (S' \cup \{x\})$,
$
 \frac{f(\{y\})}
   {w_y}
  \leq
  \frac{f(S' \cup \{x\} ) } {w(S' \cup \{x\})},
$
and
$$
f(S_2) = f(S'\cup \{x\}) +\sum_{y\in S_2 \setminus (S' \cup \{x\})} f(\{y\})
\leq f(S' \cup \{x\}) \frac{ w(S_2)}{w(S' \cup \{x\})}.
$$
%by the submodular properties of $f$ we get
%\[
%f_T(S_2) \leq f_T(S' \cup \{x\}) +
%      \sum_{y \in S \backslash (S' \cup \{x\})}
%      {f_T(S' \cup \{x,y\}) - f_T(S' \cup \{x\})}
%      \leq
%      f_T(S' \cup \{x\}) \frac{w(S_2)}{w(S' \cup \{x\})}.
%\]
%
%The first inequality follows from the linearityof $f$ and the fact that
%$S' \cup \{x\} \subseteq S_2$.
%The second inequality can be shown using \eqref{eq:relative_density} (details
%omitted).
By the linearity of $f$, we get that
%\[
$f(S') + f(\{x\})=f(S' \cup \{x\})  \geq f(S_2) \frac{L - w(T)}{w(S_2)}$.
  %\]
Since $x\in S_2 \in X_T$, we get $f(\{x\}) \leq \frac{f(T)}{|T|}$.
Hence $f(S') \geq f(S_2) \frac{L - w(T)}{w(S_2)} - \frac{f(T)}{|T|}$.

%Using the fact that $f$ is linear  and $f(\emptyset)= 0$, we
%get $f_T(S')+ f_T(\{x\}) \geq f_T(S' \cup \{x\})
% \geq  f_T(S_2) \frac{L - w(T)}{w(S_2)}$.
%Since $f_T  (x) \leq  \frac{f(T)}{|T|}$ we have that
%%\[
%$f_T(S') \geq f_T(S_2) \frac{L-w(T)}{w(S_2)} - \frac{f(T)}{|T|}$.
%                       %\]
%\hspace*{\fill} $\Box$ \vskip \belowdisplayskip
\end{proof}

Consider the iteration of Step \ref{enum:main_loop}.
in the above algorithm,
in which $T = T_2$ (assuming there are at least two elements
in the optimal solution; else $T= T_1$),
%if there is a single element in the optimal solution,
and the values of the solutions found in this iteration.
By Theorem \ref{lemma:lagrangian_relaxation}, taking
$\alpha= \frac{1}{1+r}$, one of
 the following holds:
 \begin{enumerate}
 \item
 \label{enum:s1_is_good}
 $f(S_1) \geq \frac{r}{1+r} [f(S^*)- f(T)]$
 \item
 \label{enum:s2_is_good}
 $f(S_2) \geq (1-r -\eps) [f(S^*)-f(T)]
           \frac{w(S_2)} {L-w(T)}$.
 \end{enumerate}

If \ref{enum:s1_is_good}. holds then we get
$f(S_1 \cup T) \geq f(T)+
    ( \frac{r}{1+r} - \eps )[ f(S^*) -f(T)]
 \geq (\frac{r}{1+r} - \eps)  f(S^*)$,
else we have that
%If \ref{enum:s2_is_good}. holds,
 $f(S_2) \geq (\frac{r}{1+r} -\eps) [f(S^*)-f(T)]
           \frac{w(S_2)}{L-w(T)}$,
and by Lemma \ref{lemma:enum_s'},
%%Keep for full version
$$
f(S') \geq f(S_2)\frac{L-w(T)}{w(S_2)} -
                \frac{f(T)}{|T|}\\
\geq (\frac{r}{1+r}-\eps)[f(S^*)-f(T) ] -
                \frac{f(T)}{|T|}.
$$
Hence, we have
\begin{eqnarray*}
f(S' \cup T)  &=& f(S') + f(T) \geq
  f(T)+ (\frac{r}{1+r}-\eps)[f(S^*)-f(T) ] -
                       \frac{f(T)}{|T|} \\
&=& (1- \frac{1}{k})f(T)+ (\frac{r}{1+r}-\eps)[f(S^*)-f(T) ]
\geq (\frac{r}{1+r}-\eps)f(S^*).
%&\geq& (\frac{r}{1+r}-\eps)f(S^*).
\end{eqnarray*}

The last inequality follows from choosing $k=2$, and
the fact that $\frac{1}{2} \geq \frac{r}{1+r}- \eps$.

\begin{theorem}
\label{thm:enumeration_ratio}
Algorithm \ref{alg:general} outputs an
$(\frac{r}{1+r} -\eps)$-approximation for $\Gamma$.
The number of calls to algorithm $\mathcal{A}$ is
$O( ( \log (p_{max})+ \log(L) + \log(\eps ^{-1}) )n^2)$,
where $n= \left| U \right|$ is the size of the universe of elements
for the problem $\Gamma$.
\end{theorem}

We summarize the above discussion in the next result.
\begin{corollary}
\label{cor:lagrang}
Given a subset selection problem $\Gamma$ with a linear
constraint, an algorithm $\mathcal{A}$ that
yields an $r$-approximation for $\Gamma(\lambda)$, and
$\lambda_{max}$, such $w(\mathcal{A}(\lambda_{max})) \leq L$,
there is an  $(\frac{r}{1+r} -\eps)$-approximation
algorithm, such that the number of calls of the algorithm to $\mathcal{A}$
is polynomial in $\eps$, the size of the universe, $|U|$,
and the input size.
\end{corollary}

In the Appendix we show how our technique for solving subset selection
problems with a single linear constraint can be extended to
solve such problems with multiple linear constraints, by
repetitive usage of our technique.

\subsection{Lagrangian Relaxation: Example}
\label{sec:example}
We now show the tightness of the bound in Theorem \ref{thm:enumeration_ratio}. Consider the following problem.
We are given a base set of elements $A$, where each element $a \in A$
has a profit  $p(a) \in \mathbb{N}$. Also, we have three subsets of elements $A_1,A_2,A_3 \subseteq A$,
and a bound $k >1$. We need to
select a subset $S\subseteq A$ of size at most $k$, such that
$S \subseteq A_1$, or
$S \subseteq A_2$, or $S \subseteq A_3$, and the
total profit  from elements in $S$
is, i.e., $\sum_{a\in S} p(a)$, is maximized.
The problem can be easily interpreted as a subset selection problem,
by taking the universe
to be $U= A$, the domain $X$ consists of all the subsets $S$ of
$U$, such that
$S \subseteq A_1$ or $S \subseteq A_2$, or $S \subseteq A_3$. The
weight function is $w(S) = |S|$,
with the weight bound
$L=k$, and the profit of a subset $S$ is given by $f(S)= \sum_{a\in S} p(a)$.

The Lagrangian relaxation of the problem with parameter $\lambda$ is
$\max_{S\in X} f(S)-\lambda w(S)$.
Assume that we have an algorithm $\mathcal{A}$ that
returns an $r$-approximation for the Lagrangian relaxation of the problem.

For any $\frac{1}{2}>\delta >0$ and an integer $k>\frac{1}{r}+4$,
consider the following input:
\begin{itemize}
\item
$A_1 = \{a_1,\ldots,a_{k-1},b\}$, where $p(a_i) = \frac{1}{r}$ for
$1\leq i \leq k-1$, and $p(b) = k-1$.
\item
$A_2 =\{c\}$ where $p(c) = k +\delta$.
\item
$A_3 = \{d_1,\ldots,d_\ell \}$, where
$\ell = \ceil{\frac{(1+r)(k-1)}{\delta r}}$, and $p(d_i) = 1 +\delta$
for $1 \leq  i \leq \ell$.
\item
$U= A= A_1 \cup A_2 \cup A_3$, and the set $S$ to be chosen is of size at most $k$.
\end{itemize}

Denote the profit from a subset $S \subseteq U$ by $p(S)$, and the
profit in the Lagrangian
relaxation with parameter $\lambda$ by $p_\lambda(S)$.
Clearly, the subset $S = A_1$ is an optimal solution for the problem,
of profit $p(A_1)=(k -1)\frac{1+r}{r}$.
Consider the possible solutions algorithm $\mathcal{A}$ returns
for different values of $\lambda$:
\begin{itemize}
\item
For $\lambda <1$: the profit from any subset of $A_1$ is bounded by the
original profit of
$A_1$, given by $p(A_1)=(k -1)\frac{1+r}{r}$;
the profit from the set $S=A_3$ is equal to
$p(A_3)=(1+\delta- \lambda)\ell \geq
\delta \ell \geq (k-1)\frac{(1+r)}{r}$, i.e.,
$A_3$ has higher profit than $A_1$.

\item
For $1 \leq  \lambda \leq \frac{1}{r}$: the profit from any
subset of $A_1$ is bounded by the total profit of $A_1$ (all
elements are of non-negative profit in the relaxation).
Taking the difference, we have
\begin{eqnarray*}
r  \cdot p_\lambda(A_1) - p_\lambda(A_2)  &=&
r \left( k-1-\lambda +(k-1)(\frac{1}{r} - \lambda) \right) - (k -\lambda) \\
&=& r k -r - r \lambda +k - - r \lambda k +r k -k + \lambda \\
&= & (1-\lambda)(r k -1) -r \leq 0
\end{eqnarray*}
This implies that in case the optimal set is $A_1$ (or a subset of $A_1$),
 the algorithm $\mathcal{A}$ may choose the set $A_2$.
\item
For $\lambda > \frac{1}{r}$: the maximal profit from any subset of $A_1$
is bounded in this case by $\max\{k-1-\lambda,0\}$, whereas the profit
from $A_2$ is $\max\{k-\lambda,0\}$.
\end{itemize}

From the above, we get that $\mathcal{A}$ may return a subset of $A_2$
or $A_3$ for any value of $\lambda$.
However, no combination of elements of $A_2$ and $A_3$ yields a
solution for the original problem
of profit greater than $k(1+\delta)$.
This means that, by combining the solutions returned by
the Lagrangian relaxation, one cannot achieve approximation ratio better than $\frac{k(1+\delta)}{(k-1)(1+\frac{1}{r})}=
\frac{r}{1+r}\cdot \frac{k(1+\delta)}{k-1}$.
Since $\frac{r}{1+r} \cdot \frac{k(1+\delta)}{k-1} \rightarrow \frac{r}{1+r}$
 for $(k,\delta) \rightarrow (\infty,0)$,
one cannot achieve approximation ratio better than $\frac{r}{1+r}$.

\comment{
\section{Proof of Lemma \ref{lemma:card_greedy}}
We use in the proof the next result, due to Nemhauser et al. \cite{NWF78}.
\begin{lemma}
\label{lemma:greedy_increments}
A set function $f$ over a ground set $X$ is non-decreasing submodular iff
for all $S,T \subseteq X$ $f(T) \leq f(S) + \sum_{x \in T \setminus S}
\Delta_x f$,
where  $\Delta_xf= f(S \cup \{x \})- f(S)$.
\end{lemma}

{\noindent \bf Proof of Lemma \ref{lemma:card_greedy}:}
It is easy to show (e.g., by using induction) that for every element
$s \in S_2 \backslash S'$,
\begin{equation}
\label{eq:card_relative_density}
f(S'\cup \{x\})- f(S') < \frac{f(S')}{L},
\end{equation}
since $f$ is a
non-decreasing submodular set function.

Using \eqref{eq:card_relative_density} and the fact that $f$ is submodular,
we get that
%%Conf version
\comment{
$$
f(S_2) \leq
f(S') + \sum_{s \in S_2\backslash S'} {f(S' \cup \{s\}) - f(S')}
\leq
f(S')+ (|S_2| - L ) \frac{f(S')}{L}
= f(S') \frac{|S_2|}{L}
$$
}
\begin{eqnarray*}
f(S_2) &\leq&
f(S') + \sum_{s \in S_2\backslash S'} {f(S' \cup \{s\}) - f(S')} \\
&\leq&
f(S')+ (|S_2| - L ) \frac{f(S')}{L}
= f(S') \frac{|S_2|}{L}
\end{eqnarray*}
Thus, using \eqref{eq:beta_appx}, we have that
\[
f(S') \geq f(S_2) \frac{L}{w(S_2)} \geq \beta \cdot OPT.
\]
\hspace*{\fill} $\Box$ \vskip \belowdisplayskip
}

\section{Applications to Budgeted Subset Selection}
%\section{Applying the Technique for Solving Constrained Maximization Problems}

In this section we show how
%few examples for subset selection problems
%with linear constraint for which
the technique of Section \ref{sec:rounding}
can be applied to obtain approximation algorithms
for several classic subset selection problems with a linear constraint.

\subsection{Budgeted Real Time Scheduling}
\label{sec:bba}
The budgeted real-time scheduling problem
can be interpreted as the following subset selection problem
with linear constraint. The universe $U$
consists of all instances associated with the activities
$\{ A_1, \ldots , A_m\}$.
The domain $X$ is the set of all feasible schedules;
for any $S \in X$, $f(S)$ is the profit from the instances
in $S$, and $w(S)$ is the total cost of the instances in $S$ (note
that each instance is associated with specific
time interval). The Lagrangian relaxation of this
problem is the classic \emph{interval scheduling}  problem discussed
in \cite{BB00}: the paper gives a $\frac{1}{2}$-approximation algorithm,
whose running time is $O(n \log n)$, where $n$ is the total number
of instances in the input.
%AK- text modification
Clearly,
$p_{\max}$ (as defined in \eqref{eq:def_pmax}) can be used as $\lambda_{\max}$.
By Theorem \ref{lemma:finding}, we can find two
solutions $S_1,S_2$ which satisfy the conditions of
Theorem \ref{lemma:lagrangian_relaxation}
in $O(n\log (n)\log(Lp_{\max}/ \eps))$ steps.
Then, a straightforward implementation of the
technique of Section \ref{sec:rounding_card}
yields a $\left(\frac{1}{3}-\eps \right)$-approximation algorithm
whose running time is
$O(n\log (n)\log(Lp_{\max}/\eps))$ for inputs where all instances have {\em
unit} cost.
The same approximation ratio can be obtained in
$O(n^3 \cdot \log (n) \log(Lp_{\max} / \eps))$
steps when the instances may have {\em arbitrary} costs,
using Theorem \ref{thm:enumeration_ratio} (Note that the
Lagrangian relaxation of the residual problem with respect
to a subset of elements $T$ is also an instance of the
interval scheduling problem).

Consider now the continuous case, where each instance within some activity
$A_i$, $1 \leq i \leq m$, is given by a time window.
%For the continuous case,
One way to interpret BCRS as a subset
selection problem is by setting the universe
to be all the pairs of an instance and
a time interval in which it can be scheduled.
The size of the resulting universe is
unbounded: a more careful consideration of all possible start times of any
instance yields a universe of exponential size.
The Lagrangian
relaxation of this problem is known
%referred in \cite{BB00}
as \emph{single machine scheduling with release times and deadlines},
for which a $( \frac{1}{2} -\eps)$-approximation algorithm is given
in \cite{BB00}.
Thus, we can apply our technique for
finding two solutions $S_1,S_2$ for which
Theorem \ref{lemma:lagrangian_relaxation} holds.
However, the running time of the algorithm in
Theorem \ref{thm:enumeration_ratio}
may be exponential in the input size (since the number of the
enumeration steps depends on the size of the universe,
%defined for the input,
which may be exponentially large).
%yield a polynomial time algorithm, as the number of enumeration
%steps becomes super-polynomial.
%To overcome this problem,
Thus,
we derive an approximation algorithm using the technique of Section
%we use the technique
\ref{sec:rounding_linear}.
We summarize in the next result.
%can used and implemented in polynomial time, and
%which yields
%a polynomial time algorithm with
%an approximation ratio of  $\left(\frac{1}{11} - \eps \right)$.

\begin{theorem}
\label{thm:bba}
There is a polynomial time algorithm that yields an approximation ratio of
$( \frac{1}{3} -\eps)$ for BRS, and the ratio $\left(\frac{1}{4} - \eps
\right)$ for BCRS.
\end{theorem}
Our results also hold for other budgeted variants
of problems that appear in \cite{BB00}.

\subsection{The Budgeted Generalized Assignment Problem}
\label{sec:bgap}

%The interpretation of BGAP as a subset selection
%problem was already discussed in section
% \ref{sec:rounding}.
Consider the interpretation of GBAP as a subset selection problem, as given
in Section \ref{sec:rounding}.
The Lagrangian
relaxation of BGAP (and also of the
deduced residual problems) is an instance of GAP, for which
the paper \cite{FG06} gives a $(1-e^{-1}-\eps)$-approximation
algorithm.
We can take in
%AK- text modification
Theorem \ref{lemma:finding} $\lambda_{\max}= p_{\max}$,
where $p_{\max}$ is defined by \eqref{eq:def_pmax}, and
the two solutions $S_1,S_2$
that satisfy the condition of Theorem
\ref{lemma:lagrangian_relaxation} can be found in polynomial time.
Applying the techniques of Sections \ref{sec:rounding_card} and
\ref{sec:enumeration}, we get the next result.
%\footnote{A slightly better
%ratio can be obtained using an approximation algorithm of \cite{FV06}.}

\begin{theorem}
\label{thm:bgap}
There is a polynomial time
algorithm that yields an approximation ratio of
$\frac{1- e^{-1}}{2- e^{-1}} -\eps \approx 0.387 - \eps$
for BGAP.
\end{theorem}

A slightly better approximation ratio can be obtained by using
an algorithm of \cite{FV06}.
More generally, our result holds also
for any constrained variant of the {\em separable assignment problem (SAP)}
that can be solved using a technique of \cite{FG06}.

\subsection{Budgeted Maximum Weight Independent Set}
\label{sec:bwis}

BWIS can be interpreted as the following
subset selection problem with linear constraint. The universe $U$ is
the set of all vertices in the graph, i.e., $U=V$, the domain $X$ consists of
all subsets $V'$ of $V$, such
that $V'$ is an independent set in the given graph $G$. The objective
function $f$ is $f(V')= \sum_{v\in V'} p_v$,
the weight function is $w(V')= \sum_{v \in V'} c_v$, and
the weight bound
is $L$. The Lagrangian relaxation of BWIS is an instance
of the classic WIS problem (vertices with negative profits
in the relaxation are deleted, along with their edges).
Let $|V|=n$, then by Theorem \ref{thm:enumeration_ratio},
given an approximation algorithm $\mathcal{A}$ for WIS with approximation
ratio $f(n)$, the technique of Section \ref{sec:enumeration}
yields an approximation algorithm $\mathcal{A}_{I}$ for $BWIS$, whose
approximation ratio is
%$\mathcal{A}_{I}$ achieves the approximation
$\frac{f(n)}{1+f(n)}- \eps$. The running time of $\mathcal{A}_{I}$ is
polynomial in the input size
and in $\log(1/\eps)$. If $\log(1/f(n))$ is polynomial, take
$\eps = \frac{f(n)}{n}$; the value $\log(1/\eps) = \log(1/f(n))+ \log(n)$
is polynomial in the input
size; thus, the algorithm remains polynomial. For this selection
of $\eps$, we have the following result.
\comment{
the approximation ratio of the algorithm is $g(n)=\frac{f(n)}{1+f(n)} -\frac{f(n)}{n} =\Theta(f(n))$.
%Given also that $f(n) = o(1)$, then $\lim_{n\rightarrow \infty} \frac{g(n)}{f(n)} =1$, and
this means that
the approximation ratios of $\mathcal{A}$ and $\mathcal{B}$ are asymptotically identical.
For example, this mean that using the algorithm of \cite{Ha00}, our technique
archives a $\Omega(\frac{\log^2 n}{n})$ approximation for BWIS.
}

\begin{theorem}
\label{thm:bwis}
Given an $f(n)$-approximation algorithm for WIS, where
$f(n)=o(n)$, for any $L \geq 1$ there exists a polynomial time algorithm that
outputs a $g(n)$-approximation ratio for any instance of BWIS with the budget
$L$, where $g(n)=\Theta(f(n))$, and $\lim_{n\rightarrow \infty}
\frac{g(n)}{f(n)} =1$.
\end{theorem}

This means that
the approximation ratios of $\mathcal{A}$ and $\mathcal{A}_I$ are
asymptotically the same.
Thus, for example, using the algorithm of \cite{Ha00}, our technique
achieves an $\Omega(\frac{\log^2 n}{n})$-approximation for BWIS.
Note that the above result holds for any constant number of linear
constraints added to an input for WIS, by repeatedly applying
our Lagrangian relaxation technique.

\section{Reoptimization of Subset Selection Problems}
\label{sec:reopt}
In this section we show how
our Lagrangian relaxation technique can be used to obtain $(1,\alpha)$-reapproximation algorithms for subset selection problems, where
$\alpha \in (0,1)$.
%to obtain $(1,\alpha)$-reapproximation algorithms for subset selection problems, using
% our Lagrangian relaxation technique.
To this end, we present the notion of {\em budgeted reoptimization}.
Throughout the discussion, we assume that $R(\Pi)$ is the reoptimization version of a maximization problem $\Pi$.

\subsection {\bf Budgeted Reoptimization}
\label{reopt:budgeted}
The budgeted reoptimization problem $R(\Pi, b)$  is a restricted version of $R(\Pi)$, in which we add the constraint that the
 transition cost is at most $b$, for some budget $b \geq 0$, and the transition function $\delta$.
Formally,

\begin{eqnarray}
\label{reopt_prb:general_constraint}
R(\Pi,b):& {\max}_{s \in U} & p(s) \\
%\Pi:& {maximize} & f(s) \\
\nonumber
&\mbox{subject to:} & \delta(s) \leq b.
%\\ && s \in U
\end{eqnarray}

The optimal profit for $\rpim$ is denoted $p(\co_b)$, where
$\co_b$ is the best solution that can be reached from the initial solution with transition cost at most $b$.

\begin{definition}
An algorithm ${\cal A}$ yields an $r$-approximation for $\rpim$, for $r\in (0,1]$, if
for any reoptimization input $I$, ${\cal A}$
yields a solution $s$ of profit $p(s) \geq r\cdot p(\orpim)$, and transition cost at most $b$.
\end{definition}

 \noindent{\bf Example:}
Assume that $\Pi$ is the $0$-$1$
Knapsack problem. An instance $I$ of $\Pi$ consists of a bin of capacity $B$ and $n$ items with profits $p_i \geq 0$
and weights $w_i \geq 0$, for $1\leq i\leq n$. The formal representation of the problem is $U=\{$all feasible packings of the knapsack$\}$, and $p(s)=\sum_{i\in s}{p_i}$. In the reoptimization version of the problem,
$R(\Pi)$, each instance $I$ contains also the transition cost of item $i$, given by $\delta_i \geq 0$, for $1 \leq i \leq n$.
In the budgeted reoptimization version, $\rpim$,
$U$ is restricted to contain solutions having transition cost at most $b$. Thus, $U=\{s|~ s$ is a feasible packing of
the bin, and $\delta(s)\leq b\} = \{s|~ w(s) \leq B$,
and $ \delta(s) \leq b \}$.

For various problems, $R(\Pi,b)$ satisfies the conditions of
Corollary \ref{cor:lagrang}. In particular, given a problem $\Pi$, let $\Gamma_b = \rpim$, for some $b \geq 0$, and let
$\Gamma_b(\lambda)$ be the Lagrangian relaxation of $\rpim$, i.e.,
$\Gamma_b(\lambda) = \max_{s \in U} p(s) - \lambda \cdot \delta(s)$. If $\Gamma_b(\lambda)$ yields an instance of $\Pi$ then, by Corollary \ref{cor:lagrang}, an
$r$-approximation algorithm ${\cal A}$ for $\Pi$, satisfying for certain value of $\lambda$: $w({\cal A}) \leq b$, yields
 an $(\frac{r}{r+1} - \eps)$-approximation for $\rpim$.
In the following, we show how this can be used to obtain a reaproximation algorithm for $R(\Pi)$.

\subsection {\bf Algorithm}

An instance of our reoptimization problem $R(\Pi)$ consists of a universe $U$ of $n$ items.
Each item $i$ has a non-negative profit $p_i$, and a transition cost $\delta_{i} \in \mathbb{N}$.
%For $r_1, r_2 \in (0,1]$, let $\cA$ be an
%$r_1$-approximation algorithm for $\Pi$, and $\ca_b$ an $r_2$-approximation for $R(\Pi,b)$, for any $b \geq 0$.
We give below Algorithm \ref{alg:reopt}, which uses approximation algorithms for $\Pi$ and $\rpim$ in solving
$R(\Pi)$.

\begin{center}
\fbox{
	\begin{minipage}{0.8\textwidth}
		\begin{alg} \label{alg:reopt} Reapproximating
		$R(\Pi)$ for an instance $I$
			\begin{enumerate}
			\item \label{alg:approximate}
			 	For $r_1, r_2 \in (0,1]$, let $\ca$ be an $r_1$-approximation algorithm for
				$\Pi$, and\\let $\ca_b$ be an $r_2$-approximation algorithm for
				$R(\Pi,b)$.
 			\item \label{alg:approximate}
 				Approximate $\Pi(I)$ using $\ca(I)$ :
 				$$Z \leftarrow p(\ca(I))$$
 			\item \label{alg:calculate}
				Use binary search
%\footnote{Even though $p(\ca_b(I))$ might not
 %				be monotone, We just need to find an index that
 %				changes his state}
                 to find a budget $b > 0$ satisfying:
				\begin{enumerate}
					\item \label{alg:profit}
 						$p(\ca_b(I)) \geq
 						r_2 \cdot Z$
					\item \label{alg:budget}
						$p(\ca_{b-1}(I)) < r_2 \cdot Z$
					\end{enumerate}
 			
 			\item	Return $\ca_{b}(I)$
 			\end{enumerate}
 		\end{alg}
	 \end{minipage}
 }

 \end{center}

\begin{theorem}
\label{thm:reopt}
	Let $I$ be an instance of the reoptimization problem $R(\Pi)$. For $r_1, r_2 \in (0,1]$, given
	an $r_1$-approximation algorithm $\ca$ for $\Pi$,
	and an $r_2$-approximation algorithm $\ca_b$ for $R(\Pi,b)$, for all $b \geq 0$,
	Algorithem \ref{alg:reopt} yields in polynomial time a
	$(1,r_1\cdot r_2)$-reapproximation
	for $R(\Pi)$.
\end{theorem}

Recall that ${\cal O}$ is an optimal solution for $\Pi$, and $OPT$ is a solution for $R(\Pi)$ having the minimum transition cost,
 among the solutions of
profit $p({\cal O})$.
In Section \ref{subsec:Proof}  we prove the theorem,
by showing that the solution, $S_{\ca}$, output by Algorithm \ref{alg:reopt}
has the following properties.

\begin{enumerate}
 \item [(i)]
	The total transition cost of ${\Sa}$ is at most the
	transition cost of $OPT$,
	i.e., $$\delta({\Sa}) \leq \delta(OPT)$$
 \item [(ii)]
	The profit of ${\Sa}$ satisfies
$$p(\Sa) \geq r_1 \cdot r_2 \cdot p(OPT).$$
 \end{enumerate}

Combining Theorem \ref{thm:reopt} and Corollary \ref{cor:lagrang},
we show that a wide class of reoptimization problems can be
approximated using our technique.

\begin{cor}
\label{cor:reoptlag}
Let $R(\Pi)$ be the reoptimization version of a subset selection problem $\Pi$, and let $\Gamma_b= \rpim$, for $b \geq 0$.
Denote by ${\cal A}$ an $r$-approximation algorithm for $\Pi$, for $r \in (0,1)$.
If the lagrangian relaxation of $\Gamma_b$, $\Gamma_b(\lambda)$, yields an instance of the base problem $\Pi$, for all $b \geq 0$,
then for any $\eps >0$, Algorithm \ref{alg:reopt} is a $(1,\frac{r^2}{1+r}-\varepsilon)$-reapproximation
algorithm for $R(\Pi)$.
\end{cor}

\begin{proof}
By Corollary \ref{cor:lagrang}, given $\eps' >0$, we have an $(\frac{r}{r+1} - \eps')$-approximation algorithm, ${\cal A}_b$, for $\rpim$,
for any $b \geq 0$. Thus, using Theorem \ref{thm:reopt} with algorithms ${\cal A}$ and ${\cal A}_b$, and
taking $\eps' = \frac{\varepsilon}{r}$, we obtain a
$(1,r \cdot (\frac{r}{1+r} - \frac{\varepsilon}{r}))$-reapproximation algorithm
%, i.e., a $(1,\frac{r^2}{1+r}-\varepsilon)$-reapproximation algorithm
for $R(\Pi)$.
\end{proof}

\subsection {\bf Proof of Theorem \ref{thm:reopt}}
\label{subsec:Proof}
We use in the proof the next lemmas.
\begin{lemma}
\label{lem:cost}
The solution output by Algorithm \ref{alg:reopt}
for an instance $I$
satisfies $\delta({\Sa})\leq {\delta(OPT)}$.
\end{lemma}

\begin{proof}
Let $OPT$ be a solution of minimum transition cost, among those that yield an optimal profit for $\Pi$,
and let $b^* \geq \delta(OPT)$.
By definition, $OPT$ is a valid solution of
$\Pi(R,b^*)$.
Hence, the optimal profit of $\Pi(R,b^*)$ is
$p(OPT)$, and we have
$$p(\ca_{b^*}(I)) \geq r_2 \cdot p(OPT) \geq r_2 \cdot Z.$$
It follows that, for any budget $b$ satisfying $\ca_b(I) < r_2 \cdot Z$, we have $b <\delta(OPT)$.
By Step (\ref{alg:budget}), the algorithm selects a budget $b$ such that
$p(\ca_{b-1}(I)) < r_2\cdot Z$. Hence, $b-1<\delta(OPT)$. Since
$b$ and $\delta(OPT)$ are integers, we have that $b\leq \delta(OPT)$.
\end{proof}

\begin{lemma}
\label{lem:profit}
The profit of $S_{\ca}$ satisfies
$p(S_{\ca}) \geq r_1 \cdot r_2\cdot{p(OPT)}$.
\end{lemma}

\begin{proof}
In Step (\ref{alg:profit}), Algorithm \ref{alg:reopt} selects a solution of profit at least
$r_2\cdot Z$. Also,
$Z$ is an $r_1$-approximation for $\Pi$; therefore, $Z \geq r_1\cdot p(OPT)$.
This yields the statement of the lemma.
\end{proof}

\begin{lemma}
\label{lem:time}
Algorithm \ref{alg:reopt} has polynomial running time.
\end{lemma}

\begin{proof}
The algorithm proceeds in three steps.
Step \ref{alg:approximate} is polynomial since $\ca$ runs in
polynomial time. In Step \ref{alg:calculate} we search over all budgets
$0 \leq b \leq b_{max} = \sum_{a \in I} \delta(a)$. While $b_{max}$
may be arbitrarily large,
$log(b_{max})$ is polynomial in the input size, and indeed Algorithm \ref{alg:reopt} calls $\ca_b$ $O(log(b_{max}))$ times.
\end{proof}

Combining the above lemmas, we have the statement of the theorem.
\qed

\subsection{A Reapproximation Algorithm for SRAP}
We now show how to use Algorithm \ref{alg:reopt} to
 obtain a $(1,\alpha)$-reapproximation algorithm for SRAP, for some $\alpha \in (0,1)$. Recall, that an input for the {\em real-time scheduling problem}
consists of a set
${\cal A} = \{ A_1, \ldots , A_m \}$
of \emph{activities},
where each activity consists of a set
of \emph{instances};
an instance $\cI \in A_j$ is defined
by a half open time interval $[s(\cI),e(\cI))$ in which the instance
can be scheduled ($s(\cI)$ is the start time, and $e(\cI)$ is
the end time), and a profit $p(\cI) \geq 0$.
A schedule is \emph{feasible} if it contains
at most one instance of each activity, and for any $t \geq 0$, at
most one instance is scheduled at time $t$.
The goal is to find a feasible schedule of a subset of the activities that maximizes the total profit (see, e.g., \cite{BB00}).
Let $\Pi$ be the real-time scheduling problem. Then SRAP can be cast as $R(\Pi)$, the reoptimization version of $\Pi$.

Now, given budgeted SRAP, $\Gamma_b = R(\Pi, b)$, in which the transition cost is bounded by $b$, for some $b \geq 0$,
we can write $\Gamma_b$ in the form

\[
\Gamma_b:~ \max_{S \in X} ~f(S)~~~
\]
\negA
\negA
\begin{equation}
\label{eq:subset_constraint}
\mbox{subject to:} ~ w(S) \leq b,
\end{equation}
where $X=\{\mbox{all feasible operation schedules} \}$, and $w(S)=\delta(S)$. The Lagrangian relaxation of $\Gamma_b$ is
$\Gamma_b (\lambda): ~\max_{S\in X}  {f(S) - \lambda\cdot w(S),}$.
We note that $\Gamma_b (\lambda)$ yields an instance of the real-time scheduling problem, $\Pi$.
Our base problem, $\Pi$, can be approximated within factor $1/2$ \cite{BB00}.
As shown in Section \ref{sec:bba}, budgeted real-time scheduling admits a $(\frac{1}{3} -\eps)$-approximation.
The next result follows from Theorem \ref{thm:reopt}.

\begin{theorem}
\label{thm_srap}
There is a polynomial-time $(1, 1/6 -\eps)$-reapproximation algorithm for SRAP.
\end{theorem}

\appendix
\section{Solving Multi-budgeted Subset Selection Problems}
\label{sec:multi_budgeted}
In the following we extend our technique as given in Section
\ref{sec:rounding} to handle subset selection problems with $d$ linear
constraints, for some $d >1$.
More formally, consider the problem:
\begin{eqnarray}
\label{prb:multi}
 \max_{S \in X} & f(S)~~~ \mbox{subject to:}  \\
\forall_{1\leq i \leq d}: &  w_i(S) \leq L_i, \nonumber
\end{eqnarray}
where $X$ is a lower ideal, and the functions $f$ and $w_i$ for $1 \leq i \leq d$ are
non-decreasing linear set function,
such that $f(\emptyset)=w_i(\emptyset)=0$.
This problem can be interpreted as the following subset selection problem with
a single linear constraint. Let
$X' = \left\{S \in X | \forall_{1\leq i \leq d-1}:~ w_i(S) \leq L_i \right\}$;
the linear constraint is $w_d(S) \leq L_d$, and the function $f$ remains as
defined above.
The Lagrangian relaxation of \eqref{prb:multi} has  the same form (after
removing in the relaxation elements with negative profits),
but the number of constraints is now $d-1$.
This implies that, by repeatedly applying the technique in
Section \ref{sec:enumeration},
we can obtain an approximation algorithm for \eqref{prb:multi}
from an approximation algorithm for the
non-constrained problem (in which we want to find $\max_{S\in X} f'(S)$, where
$f'$ is some linear function).
Thus, given an $r$-approximation algorithm for
the problem after ``relaxing" $d$ constraints, we derive
an $(\frac{r}{1+dr} - \eps)$-approximation algorithm
for \eqref{prb:multi}.
Note that there is a simple reduction\footnote{
Assume w.l.o.g that $L_i=1$ for every $1\leq i\leq d$, and
set the weight of an element $e$ to be $w_e = \max_{1\leq i \leq d} w_i(\{e\})$}
  of the problem in \eqref{prb:multi}
to the same problem with $d=1$, which yields
a $\frac{\rho}{d}$-approximation for \eqref{prb:multi},
given a $\rho$-approximation algorithm $\mathcal{A}$
for the problem with single constraint.
%Lagrangian relaxation, a $\frac{r}{1+(d-1)r}$-approximation
%algorithm can be obtained using $\mathcal{A}$, which
For sufficiently small $\eps > 0$, the ratio of $\frac{\rho}{1+(d-1)\rho} -
\eps$ obtained by repeatedly applying Lagrangian
relaxation and using the approximation algorithm $\mathcal{A}$ is better,
for any $\rho \in (0,1)$.

\end{document}